\documentclass[review]{elsarticle}

\usepackage{graphicx}
\usepackage{url}
\usepackage{xspace}
\usepackage{enumitem}
\usepackage{algorithmicx}
\usepackage{times}
\usepackage{algpseudocode, algorithm}
\algrenewcommand\alglinenumber[1]{\tiny #1:}
\usepackage{multicol}
\usepackage{multirow}
\usepackage{float}
\usepackage{array}
\usepackage{xcolor}
\usepackage{soul}
\usepackage{mathtools}
\usepackage{subcaption}
\usepackage{hyperref}
\usepackage{todonotes}
\usepackage{siunitx}
\usepackage{listings}
\usepackage{wrapfig}

\journal{Journal of Parallel and Distributed Computing}

\newcommand{\punt}[1]{}
\newcommand{\cmnt}[1]{}

\newtheorem{theorem}{Theorem}
\newtheorem{lemma}{Lemma}

\newcommand{\secref}[1]{Section~\ref{sec:#1}}
\newcommand{\figref}[1]{Figure~\ref{fig:#1}}
\newcommand{\tabref}[1]{Table~\ref{tab:#1}}

\newcommand{\thmref}[1]{Theorem~\ref{thm:#1}}

\newcommand{\eqnref}[1]{Eqn(\ref{eq:#1})}

\newcommand{\algoref}[1]{{Algorithm \ref{alg:#1}}}
\newcommand{\subsecref}[1]{SubSection~\ref{subsec:#1}}

\newcommand{\Lineref}[1]{Line~\ref{lin:#1}}

\newcommand{\remove}[1]{}

\newcommand{\ignore}[1]{}

\newcommand{\op} {operation}

\newcommand{\comm}[1] {\textit{committed(#1)}}
\newcommand{\aborted}[1] {\textit{aborted}(#1)}
\newcommand{\txns}[1] {\textit{txns}(#1)}

\newcommand{\evts}[1] {evts(#1)}
\newcommand{\ssch} {sub-history\xspace}

\newcommand{\valid} {valid}

\newcommand{\legal} {legal}

\newcommand{\tobj} {\emph{t-object}}

\newcommand{\opq} {opaque\xspace}

\newcommand{\opty} {opacity\xspace}

\newcommand{\sble} {serializable\xspace}
\newcommand{\sbty} {serializability\xspace}

\newcommand{\tseq} {t-sequential\xspace}

\newcommand{\mvto} {MVTO\xspace}

\newcommand{\cc} {correctness-criterion\xspace}

\newcommand{\readi} {\emph{$read_i$}\xspace}
\newcommand{\cl} {\emph{confList}\xspace}

\newcommand{\tryc} {STM.tryC\xspace}

\newcommand{\trya} {STM.tryA}

\newcommand{\cminer} {Concurrent Miner\xspace}
\newcommand{\cvalidator} {Concurrent Validator\xspace}
\newcommand{\conminer} {\emph{concurrent miner}\xspace}
\newcommand{\convalidator} {\emph{concurrent validator}\xspace}

\newcommand{\CG} {\emph{Conflict Graph}\xspace}
\newcommand{\cgraph} {\emph{conflict graph}\xspace}
\newcommand{\Cgraph} {\emph{Conflict graph}\xspace}
\newcommand{\bg} {block graph\xspace}
\newcommand{\aul} {\emph{auList[]}\xspace}
\newcommand{\confg} {\emph{confGraph}\xspace}
\newcommand{\gconfl} {\emph{getConfList}\xspace}
\newcommand{\vp} {\emph{vPred}\xspace}
\newcommand{\vc} {\emph{vCurr}\xspace}
\newcommand{\vn} {\emph{vNext}\xspace}
\newcommand{\ep} {\emph{ePred}\xspace}
\newcommand{\ec} {\emph{eCurr}\xspace}
\newcommand{\en} {\emph{eNext}\xspace}
\newcommand{\vl} {\emph{vList}\xspace}
\newcommand{\el} {\emph{eList}\xspace}
\newcommand{\inc} {\emph{inCnt}\xspace}
\newcommand{\txlog} {\emph{txlog}\xspace}
\newcommand{\adde} {\texttt{addEdge}\xspace}
\newcommand{\addv} {\texttt{addVert}\xspace}
\newcommand{\egn} {\emph{eNode}\xspace}
\newcommand{\vgn} {\emph{vNode}\xspace}
\newcommand{\enode} {\emph{eNode}\xspace}
\newcommand{\vnode} {\emph{vNode}\xspace}
\newcommand{\searchl} {\texttt{searchLocal}\xspace}
\newcommand{\searchg} {\texttt{searchGlobal}\xspace}
\newcommand{\bc} {blockchain\xspace}
\newcommand{\BC} {Blockchain\xspace}
\newcommand{\scontract} {smart contract\xspace}
\newcommand{\SContract} {smart contract\xspace}
\newcommand{\Scontract} {Smart contract\xspace}
\newcommand{\au} {atomic-unit\xspace}
\newcommand{\miner} {miner\xspace}

\newcommand{\Miner} {miner\xspace}
\newcommand{\Validator} {validator\xspace}
\newcommand{\exec} {\emph{executeCode}\xspace}

\newcommand{\nc} {\emph{nCount}\xspace}

\newcommand{\vt} {\emph{vTail}\xspace}
\newcommand{\vh} {\emph{vHead}\xspace}
\newcommand{\eh} {\emph{eHead}\xspace}
\newcommand{\et} {\emph{eTail}\xspace}
\newcommand{\tl} {\emph{thLog}\xspace}
\newcommand{\cachel} {\emph{cacheList}\xspace}

\newcommand{\begtrans} {\emph{begin}\xspace}

\newcommand{\lastw} {lastWrite}

\newcommand{\mvve} {MVVE\xspace}
\newcommand{\vie} {VE\xspace}
\newcommand{\ce} {CE\xspace}

\newcommand{\rs}{rset\xspace}
\newcommand{\ws}{wset\xspace}

\newcommand{\begt} {STM.begin\xspace}
\newcommand{\tread} {STM.read\xspace}
\newcommand{\twrite} {STM.write\xspace}

\newcommand{\commit}{\mathcal{C}}
\newcommand{\abort}{\mathcal{A}}
\newcommand{\mth} {method\xspace}

\newcommand {\incomp}[1] {#1.incomp}
\newcommand {\live}[1] {#1.live}

\newcommand{\shist}[2] {#2.subhist(#1)\xspace}

\newcommand\psnote[1]{\todo[size=\footnotesize,color=red!35]{PS: #1}}

\newcommand{\hbg}{happen-before~}

\algnewcommand\algorithmicswitch{\textbf{switch}}
\algnewcommand\algorithmiccase{\textbf{case}}
\algnewcommand\algorithmicassert{\texttt{assert}}
\algnewcommand\Assert[1]{\State \algorithmicassert(#1)}%
\algdef{SE}[SWITCH]{Switch}{EndSwitch}[1]{\algorithmicswitch\ #1\ \algorithmicdo}{\algorithmicend\ \algorithmicswitch}%
\algdef{SE}[CASE]{Case}{EndCase}[1]{\algorithmiccase\ #1}{\algorithmicend\ \algorithmiccase}%
\algtext*{EndSwitch}%
\algtext*{EndCase}%

\newcommand{\blank}[1]{\hspace*{#1}}
\algdef{SE}[DOWHILE]{Do}{doWhile}{\algorithmicdo}[1]{\algorithmicwhile\ #1}%
\newcommand{\sk}[1]{{\color{blue}#1}}

\newcommand{\specbin} {Speculative Bin\xspace}
\newcommand{\blg} {BG\xspace}

\newcommand{\sctrn} {{AU}\xspace}
\begin{document}


\begin{frontmatter}

\title{{OptSmart}: A Space Efficient \underline{Opt}imistic Concurrent Execution of \underline{Smart} Contracts\tnoteref{mytitlenote}
}

\tnotetext[mytitlenote]{A preliminary version of this paper appeared in 27$^{th}$ Euromicro International Conference On Parallel, Distributed, and Network-Based Processing (PDP\cite{Anjana+:CESC:PDP:2019}) 2019, Pavia, Italy. A poster version of this work received \textbf{Best Poster Award} in ICDCN 2019 \cite{Anjana+Poster:ICDCN:2019}.}
\tnotetext[mytitlenote2]{This manuscript covers the exhaustive related work, detailed proposed mechanism with algorithms, optimizations on the size of the block graph, rigorous correctness proof, and additional experimental evaluations with state-of-the-art.\\
$^{***}$Author sequence follows lexical order of last names.}

\author{Parwat Singh Anjana$^\dagger$, Sweta Kumari$^\ddagger$, Sathya Peri$^\dagger$, Sachin Rathor$^\dagger$,\\and Archit Somani$^\ddagger$
}
\address{$^\dagger$Department of CSE, Indian Institute of Technology Hyderabad, Telangana, India\\$^\ddagger$Department of Computer Science, Technion, Israel}



\ead{cs17mtech11014@iith.ac.in, sweta@cs.technion.ac.il, sathya\_p@cse.iith.ac.in, cs18mtech01002@iith.ac.in, archit@cs.technion.ac.il}


\begin{abstract}
Popular blockchains such as Ethereum and several others execute complex transactions in blocks through user-defined scripts known as \emph{smart contracts}. Serial execution of smart contract transactions/atomic-units (AUs) fails to harness the multiprocessing power offered by the prevalence of multi-core processors. By adding concurrency to the execution of AUs, we can achieve better efficiency and higher throughput. 

In this paper, we develop a concurrent miner that proposes a block by executing the AUs concurrently using \emph{optimistic Software Transactional Memory systems (STMs)}. It captures the independent AUs in a \emph{concurrent bin} and dependent AUs in the \emph{block graph (BG)} efficiently. Later, we propose a concurrent validator that re-executes the same AUs concurrently and deterministically using a concurrent bin followed by BG given by the miner to verify the block. We rigorously prove the correctness of concurrent execution of AUs and show significant performance gain than state-of-the-art.


\end{abstract}


\end{frontmatter}


\section{Introduction}
\label{sec:intro}
It is commonly believed that \bc{} is a revolutionary technology for doing business over the Internet. \BC is a decentralized, distributed database or ledger of records {that store the information in cryptographically linked blocks.} Cryptocurrencies such as Bitcoin \cite{Nakamoto:Bitcoin:2009} and Ethereum \cite{ethereum:url} were the first to popularize the \bc technology. \BC{s} are now considered for automating and securely storing user records such as healthcare, financial services, real estate, etc. \cmnt{\BC{s} ensure that the records are tamper-proof but publicly readable. With their amazing usefulness to revolutionize everyday life, \BC{s} are now considered for automating and securely storing user records such as healthcare, financial services, real estate, and supply chain management.}\BC network consists of multiple peers (or nodes) where peers do not necessarily trust each other. Each node maintains a copy of the distributed ledger. \emph{Clients}, users of the \bc, send requests or \emph{transactions} to the nodes of the \bc called as \emph{miners}. The miners collect multiple transactions from the clients and form a \emph{block}. Miners then propose these blocks to be added to the \bc. 
\cmnt{They follow a global consensus protocol to agree on which blocks are chosen to be added and in what order. While adding a block to the \bc, the miner incorporates the hash of the previous block into the current block. This makes it difficult to tamper with the distributed ledger. The resulting structure is in the form of a linked list or a chain of blocks and hence the name \bc.} 

The transactions sent by clients to miners are part of a larger code called as \emph{\scontract{s}} that provide several complex services such as managing the system state, ensuring rules, or credentials checking of the parties involved \cite{Dickerson+:ACSC:PODC:2017}. \Scontract{s} are like a `class' in programming languages that encapsulate data and methods which operate on the data. The data represents the state of the \scontract{} (as well as the \bc) and the \mth{s} (or functions) are the transactions that possibly can change contract state. \cmnt{A transaction invoked by a client is typically such a \mth or a collection of \mth{s} of the \scontract{s}.}Ethereum uses Solidity \cite{Solidity} while Hyperledger supports language such as Java, Golang, Node.js, etc. 

\cmnt{
	\sk{{\textbf{Listing 1: } Send function}}
	\begin{lstlisting}[escapechar=|]
	send(s_id, r_id, amount)
	{
		if(amount > bal[s_id]) |\label{line:condition}|
	  		throw;
		bal[s_id] -= amount;
		bal[r_id] += amount;
	}
	\end{lstlisting}

}

\noindent
\textbf{Motivation for Concurrent Execution of Smart Contracts: }
Dickerson et al. \cite{Dickerson+:ACSC:PODC:2017} observed that \scontract{} transactions are executed in two different contexts in Ethereum \bc{}. First, executed by miners while forming a block-- a miner selects a sequence of client requests (transactions), executes the smart contract code of these transactions in sequence, transforming the state of the associated contract in this process. The miner then stores the sequence of transactions, the resulting final state of the contracts, and the previous block hash in the block. After creating the block, the miner proposes it to be added to the blockchain through the consensus protocol. The other peers in the system, referred to as \emph{validators} in this context, validate the block proposed by the miner. They re-execute the \scontract{} transactions in the block \emph{serially} to verify the block's final states. If the final states match, then the block is accepted as valid, and the miner who appended this block is rewarded. Otherwise, the block is discarded. Thus the transactions are executed by every peer in the system. It has been observed that the validation runs several times more than the miner code \cite{Dickerson+:ACSC:PODC:2017}.

This design of \scontract{} execution is not efficient as it does not allow any concurrency. In today's world of multi-core systems, the serial execution does not utilize all the cores, resulting in lower throughput. This limitation is not specific only to Ethereum \bc{} but also applies to other popular \bc{s} as well. Higher throughput means more transaction execution per unit time, which clearly will be desired by both miners and validators.  


\ignore{
\figref{sece} illustrates the motivation behind the execution of smart contracts by concurrent miner over serial miner. Consider \figref{sece} (a) which consists of two transactions $T_1$, and $T_2$ executed by the serial miner. Here, $T_1$, and $T_2$ are writing on data-objects $x$, and $y$ respectively. Due to the  serial execution by miner, all the transactions are executing serially although they are working on different data-objects which tends to limit the throughput of miner. Whereas \figref{sece} (b) represents the concurrent execution by miner with same scenario as \figref{sece} (a) where $T_1$ and $T_2$ are running concurrently because they are working on different data-objects. Hence, concurrent execution by miner improves the throughput as compare to serial miner.
\begin{figure}
	\centerline{\scalebox{0.65}{\input{figs/sece.pdf_t}}}
	\caption{Efficient execution of smart contracts}
	\label{fig:sece}
\end{figure}
}

However, the concurrent execution of smart contract transactions is not straightforward. Because various transactions could consist of conflicting access to the shared data objects. Two contract transactions are said to be in \emph{conflict} if both of them access a shared data object, and at least one performs a write operation. Arbitrary execution of these smart contract transactions by the miners might result in the data-races leading to the inconsistent final state of the \bc. Unfortunately, it is impossible to statically identify conflicting contract transactions since contracts are developed in Turing-complete languages. The common solution for correct execution of concurrent transactions is to ensure that the execution is \emph{\sble} \cite{Papad:1979:JACM}. A usual \cc in databases, \sbty ensure that the concurrent execution is equivalent to some serial execution of the same transactions. Thus miners must ensure that their execution is \sble \cite{Dickerson+:ACSC:PODC:2017} or one of its variants as described later.

The concurrent execution of the \scontract{} transactions of a block by the validators, although highly desirable, can further complicate the situation. Suppose a miner ensures that the concurrent execution of the transactions in a block is \sble. Later a validator re-executes the same transactions concurrently. However, during the concurrent execution, the validator may execute two conflicting transactions in an order different from the miner. Thus the serialization order of the miner is different from the validator. These can result in the validator obtaining a final state different from what was obtained by the miner. Consequently, the validator may incorrectly reject the block although it is valid as depicted in \figref{conmv}.

\begin{figure}
	\centerline{\scalebox{0.4}{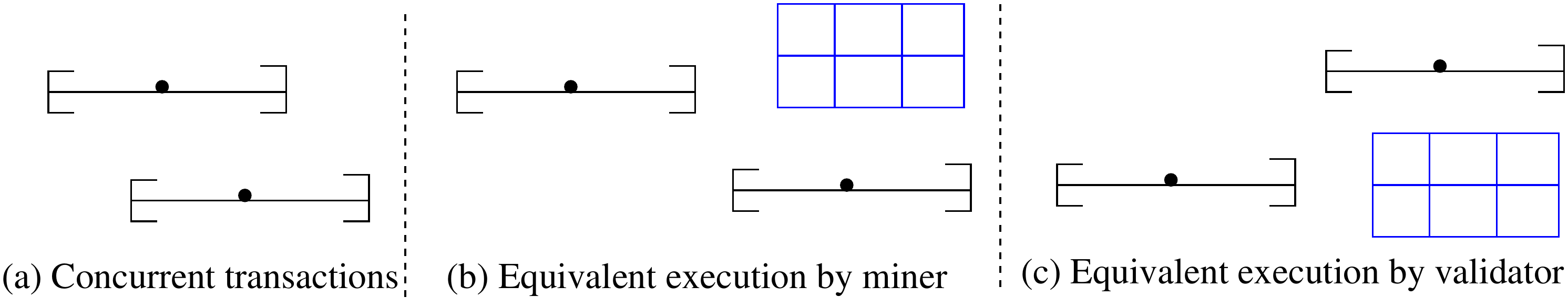}}
	\vspace{-.2cm} \caption{\small (a) consists of two concurrent conflicting transactions $T_1$ and $T_2$ working on same shared data-objects $x$ which are part of a block. (b) represents the miner's concurrent execution with an equivalent serial schedule as $T_1$, $T_2$ and final state (or FS) as 20 from the initial state (or IS) 0. Whereas (c) shows the concurrent execution by a validator with an equivalent serial schedule as $T_2$, $T_1$, and the final state as 10 from IS 0, which is different from the final state proposed by the miner. Such a situation leads to the rejection of the valid block by the validator, which is undesirable.
	}
	\label{fig:conmv}
\end{figure}

Dickerson et al. \cite{Dickerson+:ACSC:PODC:2017} identified these issues and proposed a solution for concurrent execution by both miners and validators. The miner concurrently executes block transactions using abstract locks and inverse logs to generate a serializable execution. Then, to enable correct concurrent execution by the validators, the miners provide a \emph{\hbg}graph in the block. The \hbg graph is a direct acyclic graph over all the transactions of the block. If there is a path from a transaction $T_i$ to $T_j$ then the validator has to execute $T_i$ before $T_j$. Transactions with no path between them can execute concurrently. The validator using the \hbg graph in the block executes all the transactions concurrently using the fork-join approach. This methodology ensures that the final state of the \bc generated by the miners and the validators are the same for a valid block and hence not rejected by the validators. The presence of tools such as a \hbg graph in the block provides a greater enhancement to validators to consider such blocks. It helps them execute quickly through parallelization instead of a block that does not have any parallelization tools. This fascinates the miners to provide such tools in the block for concurrent execution by the validators. 

\ignore {

\figref{cminer}, illustrates the functionality of concurrent miner, which consists of six steps. It has two or more serial miners and one concurrent miner competing to propose a block in the blockchain. Whoever will propose a block first that miner has a chance to get the strong incentive. So the challenge here is to execute the task of the miner concurrently. All the miners are getting the set of transactions from distributed shared memory. As we discussed above, the serial miner executes the transactions one after another and propose the block. Whereas concurrent miner executes the non-conflicting transactions concurrently with Transactional Memory (TM) and finally proposes a block. Complete details about the \figref{cminer} presents in the \subsecref{cminer}.

\begin{figure}
	\centering
	\captionsetup{justification=centering}
	\centerline{\scalebox{0.45}{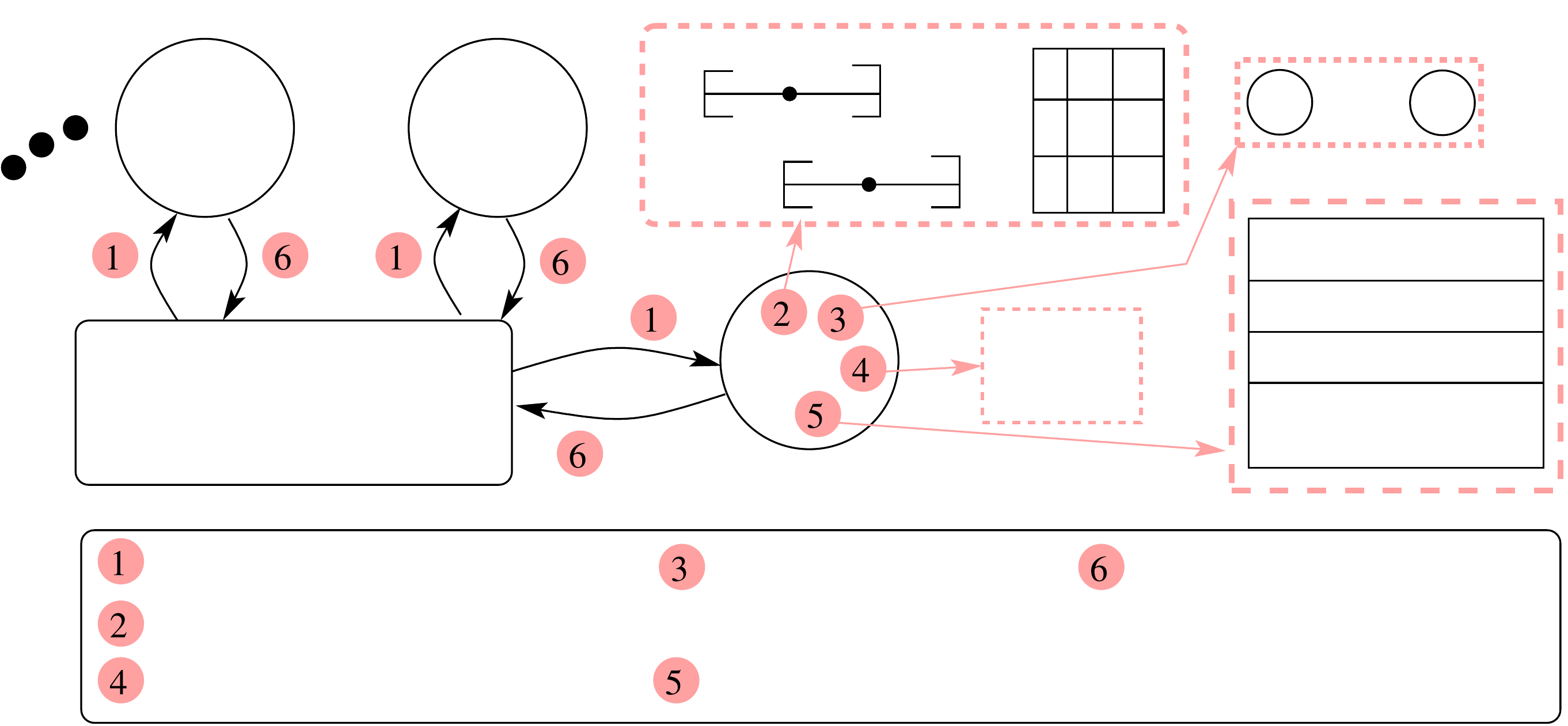}}
	\caption{Execution of Concurrent Miner}
	\label{fig:cminer}
\end{figure}

}

\vspace{1mm}
\noindent
\textbf{Proposed Solution Approach - Optimistic Concurrent Execution and Lock-Free Graph: } Dickerson et al. \cite{Dickerson+:ACSC:PODC:2017} developed a solution to the problem of concurrent miner and validators using locks and inverse logs. It is well known that locks are \emph{pessimistic} in nature. So, in this paper, we propose a \emph{novel} and \emph{efficient} framework for concurrent miner using \emph{optimistic} Software Transactional Memory Systems (STMs). STMs are suitable for the concurrent executions of transactions without worrying about consistency issues.

The requirement of the miner, is to concurrently execute the \scontract{} transactions correctly and output a graph capturing dependencies among the transactions of the block such as \hbg graph. We denote this graph as \emph{\bg} (or BG). The miner uses an optimistic STM system to execute the \scontract{} transactions concurrently in the proposed solution. Since STMs also work with transactions, we differentiate between \scontract{} transactions and STM transactions. The STM transactions invoked by an STM system is a piece of code that it tries to execute atomically even in the presence of other concurrent STM transactions. If the STM system is not able to execute it atomically, then the STM transaction is aborted. 

{The expectation of a \scontract{} transaction is that it will be executed serially. Thus, when it is executed in a concurrent setting, it is expected to execute atomically (or serialized).} 
To differentiate between \scontract{} transaction from STM transaction, we denote \scontract{} transaction as \emph{\au} (\emph{AU}) and STM transaction as \emph{transaction} in the rest of the document. Thus the miner uses the STM system to invoke a transaction for each AU. In case the transaction gets aborted, then the STM repeatedly invokes new transactions for the same AU until a transaction invocation eventually commits. 

A popular correctness guarantee provided by STM systems is \emph{\opty} \cite{GuerKap:Opacity:PPoPP:2008} which is stronger than \sbty. Opacity like \sbty requires that the concurrent execution, including the aborted transactions, be equivalent to some serial execution. This ensures that even aborted transaction reads consistent value until the point of abort. As a result, a miner using an STM does not encounter any undesirable side-effects such as crash failures, infinite loops, divide by zero, etc. STMs provide this guarantee by executing optimistically and support atomic (\opq) reads, writes on \emph{transactional objects} (or \tobj{s}). 

Due to simplicity, we have chosen two timestamp based STMs in our design: (1) \emph{Basic Timestamp Ordering} or \emph{BTO} STM \cite[Chap 4]{WeiVoss:TIS:2002:Morg}, maintains only one version for each \tobj. (2) \emph{Multi-Version Timestamp Ordering} or \emph{\mvto} STM \cite{Kumar+:MVTO:ICDCN:2014}, maintains multiple versions corresponding to each \tobj{} which further reduces the number of aborts and improves the throughput. 

The advantage of using timestamp-based STM is that the equivalent serial history is ordered based on the transactions' timestamps. Thus using the timestamps, the miner can generate the BG of the AUs. We call it as \emph{STM approach}. Dickerson et al. \cite{Dickerson+:ACSC:PODC:2017}, developed the BG in a serial manner. Saraph and Herlihy~\cite{VikramHerlihy:EmpSdy-Con:Tokenomics:2019} proposed a simple \emph{bin-based two-phase speculative} approach to execute AUs concurrently in the Ethereum blockchain without storing the BG in the block. We analyzed that the bin-based approach reduces the size of the block but fails to exploits the concurrency. We name this approach as \emph{Speculative Bin} (Spec Bin) approach. So, in our proposed approach, we combined spec bin-based approach \cite{VikramHerlihy:EmpSdy-Con:Tokenomics:2019} with the STM approach \cite{Anjana+:CESC:PDP:2019} for the optimal storage of BG in a block and exploit the concurrency. Concurrent miner generates an efficient BG in concurrent and lock-free \cite{HerlihyShavit:Progress:Opodis:2011} manner. 

The concurrent miner applies the STM approach to generate two bins while executing AUs concurrently, a concurrent bin and a sequential bin. AUs which can be executed concurrently (without any conflicts) are stored in the concurrent bin. While the AUs having conflicts are stored in a sequential bin in the BG form to record the conflicts. This combined technique reduces the size of the BG than \cite{Anjana+:CESC:PDP:2019} while storing the graph of only sequential bin \sctrn{s} instead of all \sctrn{s}.

We propose a concurrent validator that creates multiple threads. Each of these threads parses the concurrent bin followed by efficient BG provided by the concurrent miner and re-execute the AUs for validation. The BG consists of only dependent AUs. Each validator thread claims a node that does not have any dependency, i.e., a node without any incoming edges by marking it. After that, it executes the corresponding AUs deterministically. Since the threads execute only those nodes with no incoming edges, the concurrently executing AUs will not have any conflicts. Hence the validator threads need not have to worry about synchronization issues. We denote this approach adopted by the validator as a \emph{decentralized approach} as the multiple threads are working on BG concurrently in the absence of a master thread. 

The approach adopted by Dickerson et al. \cite{Dickerson+:ACSC:PODC:2017}, works on \emph{fork-join} in which a master thread allocates different tasks to slave threads. The master thread identifies AUs that do not have any incoming dependencies in the BG and allocates them to different slave threads. In this paper, we compare the performance of both these approaches with the serial validator. 

\noindent
\textbf{The significant contributions of the paper are as follows:}
\begin{itemize}[noitemsep]
\item Introduce a novel way to execute the AUs by concurrent miner using \emph{optimistic} STMs (\secref{pm}). 
We implement the concurrent miner using BTO and MVTO STM, but it is generic to any STM protocol.

\item We propose a \emph{lock-free} and concurrent graph library to generate the \emph{efficient} BG which contains only dependent \au{s} and optimize the size of the block than \cite{Anjana+:CESC:PDP:2019} (see \secref{pm}).
\item We propose concurrent validator that re-executes the AUs deterministically and efficiently with the help of \emph{concurrent bin} followed by \emph{efficient} BG given by concurrent miner (see \secref{pm}).

\item To make our proposed approach storage optimal and efficient, we have optimized the BG size (see \secref{pm}). 


\item We rigorously prove that the concurrent miner and validator satisfies correctness criterion as \emph{opacity} (see \secref{correctness}). 

\item We achieve $4.49\times$ and $5.21\times$ average speedups for optimized concurrent miner using BTO and MVTO STM protocol, respectively. Optimized concurrent BTO and MVTO decentralized validator outperform average $7.68\times$ and $8.60\times$ than serial validator, respectively (\secref{opt-result}). 

\end{itemize}

\secref{relatedwork} presents the related work on concurrent execution of smart contract transactions. While, \secref{model} includes the notions related to STMs and execution model used in the paper. The conclusion with several future directions is presented in \secref{con}.

\section{Related Work}
\label{sec:relatedwork}
This section presents the related work on concurrent execution on \bc{s} in line with the proposed approach.

{The interpretation of \emph{Blockchain} was introduced by Satoshi Nakamoto in 2009 as Bitcoin~\cite{Nakamoto:Bitcoin:2009} to perform electronic transactions without third party interference. Nick Szabo~\cite{Nick:PublicN:journals:1997} introduced \SContract{s} in 1997, adopted by Ethereum \bc{} in 2015 to expand \bc{} functionalities beyond financial transactions (cryptocurrencies).
} A smart contract is an interface to reduce the computational transaction cost and provides secure relationships on distributed networks. There exist several papers \cite{Luu+:DIC:CCS:2015, Delmolino+:SSTCSC:FC:2016, Luu+:MSC:CCS:2016} in the literature that works on the safety and security concern of smart contracts, which is out of the scope of this paper. We mainly focus on the concurrent execution of AUs. {A concise summary of closely related works is given in \tabref{relatedwork}.}

\begin{table*}[!tb]
    \caption{Related Work Summary}\vspace{-.25cm}
    \centering
    \label{tab:relatedwork}
    \resizebox{1\columnwidth}{!}{%
        \begin{tabular}{|c|c|c|c|c|c|c|}
            \hline
            \textbf{ } 
            & \textbf{Miner Approach}
            & \textbf{Locks}
            & \textbf{Require Block Graph}
            & \textbf{Validator Approach}

            & \textbf{\begin{tabular}[c]{@{}c@{}}Blockchain Type\end{tabular}}

            \\ \hline\hline

            Dickerson et al.~\cite{Dickerson+:ACSC:PODC:2017}
            & Pessimistic ScalaSTM 
            & Yes 
            & Yes 
            & Fork-join
            & Permissionless
            \\ 
            
            Zhang and Zhang~\cite{ZangandZang:ECSC:WBD:2018} 
            & - 
            & - 
            & Read, Write Set 
            & MVTO Approach 
            & Permissionless
            \\ 
            
            Anjana et al.~\cite{Anjana+:CESC:PDP:2019} 
            & Optimistic RWSTM 
            & No 
            & Yes 
            & Decentralized 
            & Permissionless
            \\ 

            Amiri et al.~\cite{amiri2019parblockchain} 
            & Static Analysis
            & -
            & Yes
            & -
            & Permissioned
            \\ 
            
            Saraph and Herlihy~\cite{VikramHerlihy:EmpSdy-Con:Tokenomics:2019} 
            & Bin-based Approach 
            & Yes 
            & No 
            & Bin-based 
            & Permissionless
            \\ 
            
            Anjana et al.~\cite{anjana:ObjSC:Netys:2020} 
            & Optimistic ObjectSTM 
            & No 
            & Yes 
            & Decentralized 
            & Permissionless
            \\ 

            \textbf{Proposed Approach} 
            & \textbf{Bin+Optimistic RWSTM} 
            & \textbf{No} 
            & \textbf{No (if no dependencies) / Yes} 
            & \textbf{Decentralized} 
            & \textbf{Permissionless}
            \\ \hline

        \end{tabular}%
    }
\end{table*}


Dickerson et al.~\cite{Dickerson+:ACSC:PODC:2017} introduced concurrent executions of AUs in the \bc{}. They observed that miners and validators could execute \sctrn{s} simultaneously to exploit concurrency offered by ubiquitous multi-core processors. {The approach of this work is given in \secref{intro}.}

Zhang and Zhang~\cite{ZangandZang:ECSC:WBD:2018} proposed a concurrent miner using a pessimistic concurrency control protocol, which delays the read until the corresponding writes to commit and ensures a conflict-serializable schedule. {The proposed concurrent validator uses MVTO protocol to execute transactions concurrently using the write sets provided by the concurrent miner in the block.}
Anjana et al.~\cite{Anjana+:CESC:PDP:2019} proposed optimistic \emph{Read-Write STM }(RWSTM) using BTO and MVTO based protocols. The timestamp-based protocols are used to identify the conflicts between \sctrn{s}. The miner executes the \sctrn{s} using RWSTM and constructs the \blg{} dynamically at the runtime using the timestamps. Later, a concurrent \emph{Decentralized Validator} (Dec-Validator) executes the \sctrn{s} in the block in a decentralized manner. The Decentralized Validator is efficient than the Fork-Join Validator since there is no master validator thread to allocate the \sctrn{s} to the slave validator threads to execute. Instead, all the validator threads identify the source vertex (a vertex with indegree 0) in the \blg{} independently and claim the source node to execute the corresponding \sctrn{}. 

Amiri et al.~\cite{amiri2019parblockchain} proposed \emph{ParBlockchain}-- an 
approach for concurrent execution of transactions in the block for permissioned blockchain. They developed an \emph{OXII paradigm}\footnote{A paradigm in which transactions are first ordered for concurrent execution then executed by both miners and validators~\cite{amiri2019parblockchain}.} to support distributed applications. The OXII paradigm orders the block transactions based on the agreement between the orderer nodes using static analysis or speculative execution to obtain the read-set and write-set of each transaction, then generates the \blg{} and constructs the block. The executors from respective applications (similar to the executors in fabric channels) execute the transactions concurrently and then validate them by re-executing the transaction. So, the nodes of the ParBlockchain execute the transactions in two phases using the OXII paradigm. A block with \blg{} based on the transaction conflicts is generated in the first phase, known as the \emph{ordering phase}. The second phase, known as the \emph{execution phase}, executes the block transactions concurrently using the \blg{} appended with block.


Saraph and Herlihy~\cite{VikramHerlihy:EmpSdy-Con:Tokenomics:2019} proposed a simple \emph{bin-based two-phase speculative} approach to execute \sctrn{s} concurrently in the Ethereum blockchain. They empirically validated the possible benefit of their approach by evaluating it on historical transactions from the Ethereum. In the first phase, the miner uses locks and executes \sctrn{s} in a block concurrently by rolling back those \sctrn{s} that lead to the conflict(s). All the aborted \sctrn{s} are then kept into a sequential bin and executed in the second phase sequentially. The miner gives concurrent and sequential bin hints in the block to the validator to execute the same schedule as executed by the miner. The validator executes the concurrent bin \sctrn{s} concurrently while executes the sequential bin \sctrn{s} sequentially. Instead of \blg{}, giving hints about bins takes less space. However, it does not harness the maximum concurrency available within the block.  
\cmnt{We name this approach as \emph{Speculative Bin} (Spec Bin) approach. }


Later, Anjana et al.~\cite{anjana:ObjSC:Netys:2020} proposed an approach that uses optimistic single-version and multi-version \emph{Object-based STMs (OSTMs)} for the concurrent execution of \sctrn{s} by the miner. The OSTMs operate at a higher (object) level rather than page (read-write) level and constructs the BG. 
However, the \blg{} is still quite significantly large in the existing approaches and needs higher bandwidth to broadcast such a large block for validation. 

In contrast, we propose an efficient framework for concurrent execution of the AUs using optimistic STMs. 
We combine the benefits of both Spec Bin-based and STM-based approaches to optimize the storage aspects (efficient storage optimal \blg{}), which further improves the performance. Due to its optimistic nature, the updates made by a transaction will be visible to shared memory only on commit; hence, rollback is not required. Our approach ensures correctness criteria as opacity~\cite{GuerKap:Opacity:PPoPP:2008}. The proposed approach gives better speedup over state-of-the-art and serial execution of \sctrn{s}.

\section{System Model}
\label{sec:model}


{In this section, we will present the notions related to STMs and the execution model used in the proposed approach.}

Following~\cite{tm-book,KuznetsovPeri:Non-interference:TCS:2017}, we assume a system of $n$ processes/threads, $p_1,\ldots,p_n$ that access a collection of \emph{transactional objects} or \tobj{s} via atomic \emph{transactions}. Each transaction has a unique identifier. Within a transaction, processes can perform \emph{transactional operations or \mth{s}}: 
\begin{itemize}
	\item \texttt{\begt{()}}-- begins a transaction. 
	\item \texttt{\twrite}$(x,v)$ (or $w(x, v)$)-- updates a \tobj{} $x$ with value $v$ in its local memory. 
	\item \texttt{\tread}$(x, v)$ (or $r(x, v)$)-- tries to read  $x$ and returns value as $v$.
	\item \texttt{\tryc}$()$-- tries to commit the transaction and returns $commit$ (or $\commit$) if succeeds. 
	\item \texttt{\trya}$()$-- aborts the transaction and returns $\abort$. 
\end{itemize}
Operations \texttt{\tread{()}} and \texttt{\tryc}$()$ may return $\abort{}$. 
Transaction $T_i$ starts with the first operation and completes when any of its operations return $\abort$ or $\commit$. 
For a transaction $T_k$, we denote all the \tobj{s} accessed by its read \op{s} and write operations as $\rs_k$ and $\ws_k$, respectively.  We denote all the \op{s} of a transaction $T_k$ as $\evts{T_k}$ or $evts_k$.

\vspace{.2cm}
\noindent
\textbf{History:}
A \emph{history} is a sequence of \emph{events}, i.e., a sequence of 
invocations and responses of transactional operations. The collection of events is denoted as $\evts{H}$. For simplicity, we consider \emph{sequential} histories, i.e., the invocation of each transactional operation is immediately followed by a matching response. Therefore, we treat each transactional operation as one atomic event and let $<_H$ denote the total order on the transactional operations incurred by $H$. 
We identify a history $H$ as tuple $\langle \evts{H},<_H \rangle$. 

Further, we consider \emph{well-formed} histories, i.e., no transaction of a process begins before the previous transaction invocation has completed (either $commits$ or $aborts$). We also assume that every history has an initial \emph{committed} transaction $T_0$ that initializes all the t-objects with value $0$. The set of transactions that appear in $H$ is denoted by $\txns{H}$. The set of \emph{committed} (resp., \emph{aborted}) transactions in $H$ is denoted by $\comm{H}$ (resp., $\aborted{H}$). The set of \emph{incomplete} or \emph{live} transactions in $H$ is denoted by $\incomp{H} = \live{H} = (\txns{H}-\comm{H}-\aborted{H})$. 

We construct a \emph{complete history} of $H$, denoted as $\overline{H}$, by inserting $\trya_k(\abort)$ immediately after the last event of every transaction $T_k\in \live{H}$. But for $\tryc_i$ of transaction $T_i$, if it released the lock on first \tobj{} successfully that means updates made by $T_i$ is consistent so, $T_i$ will immediately return commit.

\cmnt{
\noindent
\textbf{Sub-history:} A \textit{sub-history} ($SH$) of a history ($H$) 
denoted as the tuple $\langle \evts{SH},$ $<_{SH}\rangle$ and is defined as: 
(1) $<_{SH} \subseteq <_{H}$; (2) $\evts{SH} \subseteq \evts{H}$; (3) If an 
event of a transaction $T_k\in\txns{H}$ is in $SH$ then all the events of $T_k$ 
in $H$ should also be in $SH$. 

For a history $H$, let $R$ be a subset of $\txns{H}$. Then $\shist{R}{H}$ denotes  the \ssch{} of $H$ that is formed  from the \op{s} in $R$.
}

\vspace{.2cm}
\noindent
\textbf{\textit{Transaction Real-Time and Conflict order:}} For two transactions $T_k,T_m \in \txns{H}$, we say that  $T_k$ \emph{precedes} $T_m$ in the \emph{real-time order} of $H$, denoted as $T_k\prec_H^{RT} T_m$, if $T_k$ is complete in $H$ and the last event of $T_k$ precedes the first event of $T_m$ in $H$. If neither $T_k \prec_H^{RT} T_m$ nor $T_m \prec_H^{RT} T_k$, then $T_k$ and $T_m$ \emph{overlap} in $H$. We say that a history is \emph{serial} (or \emph{\tseq}) if all the transactions are ordered by real-time order. 
We say that $T_k, T_m$ are in conflict, denoted as $T_k\prec_H^{Conf} T_m$, if

(1) $\tryc_k()<_H \tryc_m()$ and $wset(T_k) \cap wset(T_m) \neq\emptyset$; 

(2) $\tryc_k()<_H r_m(x,v)$, $x \in wset(T_k)$ and $v \neq \abort$;

(3) $r_k(x,v)<_H \tryc_m()$, $x\in wset(T_m)$ and $v \neq \abort$. 

Thus, it can be seen that the conflict order is defined only on \op{s} that have successfully executed. We denote the corresponding \op{s} as conflicting.

\vspace{.2cm}
\noindent
\textbf{Valid and Legal histories:} A successful read $r_k(x, v)$ (i.e., $v \neq \abort$)  in a history $H$ is said to be \emph{\valid} if there exist a transaction $T_j$ that wrote $v$ to $x$ and \emph{committed} before $r_k(x,v)$. 
History $H$ is \valid{} if all its successful read \op{s} are \valid. 

We define $r_k(x, v)$'s \textit{\lastw{}} as the latest commit event $\commit_i$ preceding $r_k(x, v)$ in $H$ such that $x\in wset_i$ ($T_i$ can also be $T_0$). A successful read \op{} $r_k(x, v)$ (i.e., $v \neq \abort$), is said to be \emph{\legal{}} if the transaction containing $r_k$'s \lastw{} also writes $v$ onto $x$. 
 The history $H$ is \legal{} if all its successful read \op{s} are \legal.  From the definitions we get that if $H$ is \legal{} then it is also \valid.

%

\vspace{.2cm}
\noindent
\textbf{Notions of Equivalence:} Two histories $H$ and $H'$ are \emph{equivalent} if they have the same set of events. We say two histories $H, H'$ are \emph{multi-version view equivalent} \cite[Chap. 5]{WeiVoss:TIS:2002:Morg} or \emph{\mvve} if 

(1) $H, H'$ are valid histories and 

(2) $H$ is equivalent to $H'$. 

\noindent
Two histories $H, H'$ are \emph{view equivalent} \cite[Chap. 3]{WeiVoss:TIS:2002:Morg} or \emph{\vie} if 

(1) $H, H'$ are legal histories and 

(2) $H$ is equivalent to $H'$. By restricting to \legal{} histories, view equivalence does not use multi-versions. 

\noindent
Two histories $H, H'$ are \emph{conflict equivalent} \cite[Chap. 3]{WeiVoss:TIS:2002:Morg} or \emph{\ce} if 

(1) $H, H'$ are legal histories and 

(2) conflict in $H, H'$ are the same, i.e., $conf(H) = conf(H')$.

Conflict equivalence like view equivalence does not use multi-versions and restricts itself to \legal{} histories. 

\vspace{.2cm}
\noindent
\textbf{VSR, MVSR, and CSR:} A history $H$ is said to VSR (or View Serializable) \cite[Chap. 3]{WeiVoss:TIS:2002:Morg}, if there exist a serial history $S$ such that $S$ is view equivalent to $H$. But this notion considers only single-version corresponding to each \tobj{}.

MVSR (or Multi-Version View Serializable) maintains multiple version corresponding to each \tobj. A history $H$ is said to MVSR \cite[Chap. 5]{WeiVoss:TIS:2002:Morg}, if there exist a serial history $S$ such that $S$ is multi-version view equivalent to $H$. It can be proved that verifying the membership of VSR as well as MVSR in databases is NP-Complete \cite{Papad:1979:JACM}. To circumvent this issue, researchers in databases have identified an efficient sub-class of VSR, called CSR based
on the notion of conflicts. The membership of CSR can be verified in polynomial time using conflict graph characterization. 

A history $H$ is said to CSR (or Conflict Serializable) \cite[Chap. 3]{WeiVoss:TIS:2002:Morg}, if there exist a serial history $S$ such that $S$ is conflict equivalent to $H$.

\vspace{.2cm}
\noindent
\textbf{Serializability and Opacity:} 
Serializability \cite{Papad:1979:JACM} is a commonly used criterion in databases. But it is not suitable for STMs as it does not consider the correctness of \emph{aborted} transactions as shown by Guerraoui and Kapalka \cite{GuerKap:Opacity:PPoPP:2008}. Opacity, on the other hand, considers the correctness of \emph{aborted} transactions as well.

\noindent
A history $H$ is said to be \textit{opaque} \cite{GuerKap:Opacity:PPoPP:2008,tm-book} if it is \valid{} and there exists a t-sequential legal history $S$ such that 

(1) $S$ is equivalent to complete history $\overline{H}$ and 

(2) $S$ respects $\prec_{H}^{RT}$, i.e., $\prec_{H}^{RT} \subset \prec_{S}^{RT}$. 

By requiring $S$ being equivalent to $\overline{H}$, opacity treats all the incomplete transactions as aborted. 
Similar to view-serializability, verifying the membership of \opty is NP-Complete \cite{Papad:1979:JACM}. To address this issue, researchers have proposed another popular correctness-criterion \emph{co-opacity} whose membership is polynomial time verifiable.

\vspace{.2cm}
\noindent
\textbf{Co-opacity:} A history $H$ is said to be \textit{co-opaque} \cite{KuznetsovPeri:Non-interference:TCS:2017} if it is \valid{} and there exists a t-sequential legal history $S$ such that 

(1) $S$ is equivalent to complete history $\overline{H}$ and 

(2) $S$ respects $\prec_{H}^{RT}$, i.e., $\prec_{H}^{RT} \subset \prec_{S}^{RT}$. 

(3) S preserves conflicts (i.e. $\prec^{Conf}_{H}\subseteq\prec^{Conf}_{S}$). 

\cmnt{
Along same lines, a \valid{} history $H$ is said to be \textit{strictly serializable} if $\shist{\comm{H}}{H}$ is opaque. Unlike opacity, strict serializability does not include aborted or incomplete transactions in the global serialization order. An opaque history $H$ is also strictly serializable: a serialization of $\shist{\comm{H}}{H}$ is simply the subsequence of a serialization of $H$ that only contains transactions in $\comm{H}$. 
}

\cmnt{
\noindent 
\textbf{VSR, MVSR, and CSR:} VSR (or View Serializability) is a correctness-criterion similar to opacity but does not consider correctness of aborted transactions. When the protocol (such as MVTO) maintains multiple versions corresponding to each \tobj{} then a commonly used correctness criterion in databases is MVSR (or Multi-Version View Serializability). It can be proved that verifying the membership of VSR as well as MVSR in databases is NP-Complete \cite{Papad:1979:JACM}. To circumvent this issue, researchers in databases have identified an efficient sub-class of VSR, called CSR (or conflict-serializability), based
on the notion of conflicts. The membership of CSR can be verified in polynomial time using conflict graph characterization. 
}

\vspace{.2cm}
\noindent
\textbf{Linearizability:} A history $H$ is linearizable \cite{HerlihyandWing:1990:LCC:ACM} if 

(1) The invocation and response events can be reordered to get a valid sequential history.

(2) The generated sequential history satisfies the object’s sequential specification. 

(3) If a response event precedes an invocation event in the original history, then this should be preserved in the sequential reordering.

\vspace{.2cm}
\noindent
\textbf{Lock Freedom:} An algorithm is said to be lock-free \cite{HerlihyShavit:Progress:Opodis:2011} if the program threads are run for a sufficiently long time, at least one of the threads makes progress. It allows individual threads to starve but guarantees system-wide throughput.

\ignore{


\begin{enumerate}
	\item BTO 
	\begin{itemize}
		\item Conflicts.
	\end{itemize}
	\item MVTO 
	\begin{itemize}
		\item Conflicts.
	\end{itemize}
\end{enumerate}
}
\section{Proposed Mechanism}
\label{sec:pm}
This section presents the methods of lock-free concurrent block graph library followed by concurrent execution of AUs by miner and validator.
\subsection{Lock-free Concurrent Block Graph}
\label{subsec:bg}
\noindent
\textbf{Data Structure of Lock-free Concurrent Block Graph:} We use the \emph{adjacency list} to maintain the block graph BG(V, E), as shown in \figref{confg} (a). Where V is a set of vertices (or \vnode{s}) which are stored in the vertex list (or \vl{}) in increasing order of timestamp between two sentinel node \vh{} (-$\infty$) and \vt{} (+$\infty$). Each vertex node (or \vnode) contains $\langle ts = i, AU_{id} = id, \inc{} = 0, \vn{} = nil, \en{} = nil\rangle$. Where $i$ is a unique timestamp (or $ts$) of transactions $T_i$. $AU_{id}$ is the $id$ of a \au{} executed by transaction $T_i$. To maintain the indegree count of each \vnode{}, we initialize \inc{} as 0. \vn{} and \en{} initialize as $nil$. 

\begin{figure}
	\centering
		\scalebox{.42}{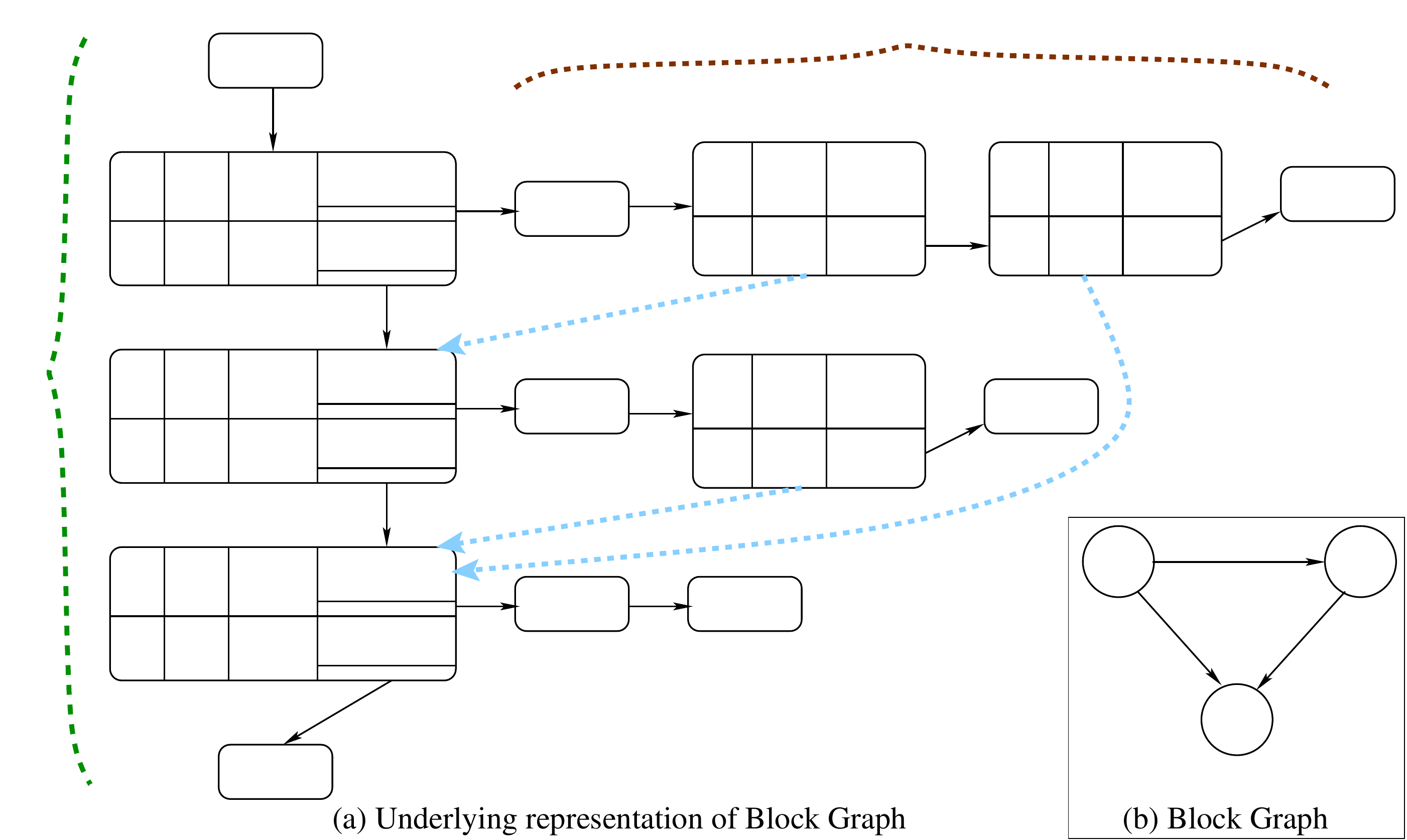}
		\centering
		\caption{Pictorial representation of Block Graph}
		\label{fig:confg}
\end{figure}

\cmnt{
\begin{figure*}[!htb]
   \begin{minipage}{0.49\textwidth}
     \centering
     \scalebox{.31}{\input{figs/graphs.pdf_t}}\vspace{-.3cm}
     \caption{Pictorial representation of Block Graph}\label{fig:confg}
   \end{minipage}
   \begin{minipage}{0.49\textwidth}
     \centering
     \scalebox{.31}{\input{figs/cminer.pdf_t}}\vspace{-.3cm}
     \caption{Execution of Concurrent Miner}\label{fig:cminer}
   \end{minipage}
\end{figure*}
}

\cmnt{
\begin{figure*}[!htb]
     \centering
     \scalebox{.31}{\input{figs/graphs.pdf_t}}\vspace{-.3cm}
     \caption{Pictorial representation of Block Graph}\label{fig:confg}
\vspace{.35cm}
     \centering
     \scalebox{.31}{\input{figs/cminer.pdf_t}}\vspace{-.3cm}
     \caption{Execution of Concurrent Miner}\label{fig:cminer}
\end{figure*}

\begin{figure*}
	\centering
	\begin{minipage}[b]{0.3\textwidth}
		\scalebox{.31}{\input{figs/graphs.pdf_t}}
		\centering
		\caption{Pictorial representation of Block Graph}
		\label{fig:confg}
	\end{minipage}
	\hfill
	\begin{minipage}[b]{0.5\textwidth}
		\scalebox{.32}{\input{figs/cminer.pdf_t}}
		\centering
		\caption{Execution of Concurrent Miner}
		\label{fig:cminer}
	\end{minipage}   
\end{figure*}
}

While E is a set of edges which maintains all conflicts of \vnode{} in the edge list (or \el), as shown in \figref{confg} (a). \el{} stores \enode{s} (or conflicting transaction nodes, say $T_j$) in increasing order of timestamp between two sentinel nodes \eh{} (-$\infty$) and \et{} (+$\infty$). Edge node (or \enode{}) contains $\langle$\emph{ts = j, vref}, \en{} = $nil$$\rangle$. Here, $j$ is a unique timestamp (or $ts$) of \emph{committed} transaction $T_j$ having a conflict with $T_i$ and $ts(T_i)$ is less than $ts(T_j)$. We add conflicting edges from lower timestamp to higher timestamp transactions to maintain the acyclicity in the BG i.e., conflict edge is from $T_i$ to $T_j$ in the BG. \figref{confg} (b) illustrates this using three transactions with timestamp 0, 5, and 10, which maintain the acyclicity while adding an edge from lower to higher timestamp. To make it search efficient, \emph{vertex node reference (or vref)} keeps the reference of its own vertex which is present in the \vl and \en{} initializes as $nil$. 

The block graph (BG) generated by the concurrent miner helps to execute the validator concurrently and deterministically through lock-free graph library methods. Lock-free graph library consists of five methods as follows: \texttt{addVert(), addEdge(), searchLocal(), searchGlobal()} and \texttt{decInCount()}. 

\noindent
\textbf{Lock-free Graph Library Methods Accessed by Concurrent Miner:} The concurrent miner uses \texttt{addVert()} and \texttt{addEdge()} methods of lock-free graph library to build a BG. When concurrent miner wants to add a node in the BG, it first calls the \texttt{addVert()} method. The  \texttt{addVert()} method identifies the correct location of that node (or \vgn{}) in the \vl{} at \Lineref{addv2}. If \vgn{} is not part of \vl{}, it creates the node and adds it into \vl{} at \Lineref{addv5} in a lock-free manner using atomic compare and swap (CAS) operation. Otherwise, \vgn{} is already present in \vl{} at \Lineref{addv10}.\vspace{.2cm}

\begin{algorithm}[H]
	\scriptsize
	\label{alg:cg} 
	\caption{BG(\emph{vNode}, STM): It generates a BG for all the atomic-unit nodes.}
	\setlength{\multicolsep}{0pt}
	\begin{multicols}{2}
		\begin{algorithmic}[1]
			\makeatletter\setcounter{ALG@line}{0}\makeatother
			\Procedure{BG(\emph{vNode}, STM)}{} \label{lin:cg1}
			\State /*Get the \cl{} of transaction $T_i$ from STM*/\label{lin:cg2}
			\State clist $\gets$ STM.\gconfl(\emph{vNode}.$ts_i$);\label{lin:cg3}
			\State /*$T_i$ conflicts with $T_j$ and $T_j$ existes in conflict list \blank{.3cm} of $T_i$*/\label{lin:cg4}
			\ForAll{($ts_j$ $\in$ clist)}\label{lin:cg5}
			\State \addv(\emph{$ts_j$}); \label{lin:cg6}
			\State \addv(\emph{vNode.$ts_i$});\label{lin:cg7}
			\If{($ts_j$  $<$ \emph{vNode}.$ts_i$)}\label{lin:cg8}
			\State \adde($ts_j$, \emph{vNode}.$ts_i$);\label{lin:cg9}
			\Else\label{lin:cg10}
			\State \adde(\emph{vNode}.$ts_i$, $ts_j$);\label{lin:cg11}
			\EndIf  \label{lin:cg12}
			\EndFor\label{lin:cg13}
			\EndProcedure\label{lin:cg14}
		\end{algorithmic}
	\end{multicols}
\end{algorithm}

\begin{algorithm}[H]
	\scriptsize
	\label{alg:addv} 	
	\caption{\emph{\addv}{($ts_i$)}: It adds the vertex in the BG for $T_i$. 
	} 
	\setlength{\multicolsep}{0pt}
	\begin{multicols}{2}
		\begin{algorithmic}[1]
			\makeatletter\setcounter{ALG@line}{14}\makeatother
			\Procedure{\emph{\addv}{($ts_i$)}}{} \label{lin:addv1}
			\State Identify $\langle$\vp, \vc{}$\rangle$ of \vgn{} of $ts_i$ in \vl{};\label{lin:addv2}
			\If{(\vc.$ts_i$ $\neq$ \vgn.$ts_i$)}\label{lin:addv3}
			\State Create new Graph Node (\vgn) of $ts_i$ in \vl{};\label{lin:addv4}
			\If{(\vp.\vn.CAS(\vc, \vgn))}\label{lin:addv5}
			
			\State return$\langle$\emph{Vertex added}$\rangle$; 
			\label{lin:addv6}
			\EndIf\label{lin:addv7}
			\State goto \Lineref{addv2}; /*Start with the \vp{} to identify \blank{.7cm} the new $\langle$\vp, \vc{}$\rangle$*/ \label{lin:addv8}
			\Else\label{lin:addv9}
			\State return$\langle$\emph{Vertex already present}$\rangle$; 
			\label{lin:addv10}
			\EndIf\label{lin:addv11}
			\EndProcedure\label{lin:addv12}
		\end{algorithmic}
	\end{multicols}
\end{algorithm}

\begin{algorithm}[H]
	\scriptsize
	\label{alg:adde} 	
	\caption{\emph{\adde{(fromNode, toNode)}}: It adds an  edge from \emph{fromNode} to \emph{toNode}.}
	\setlength{\multicolsep}{0pt}
	\begin{multicols}{2}	
		\begin{algorithmic}[1]
			\makeatletter\setcounter{ALG@line}{26}\makeatother
			\Procedure{\emph{\adde{(fromNode, toNode)}}}{}\label{lin:adde1}
			\State Identify the $\langle$\ep, \ec{}$\rangle$ of \emph{toNode} in \el{} of \blank{.3cm} the \emph{fromNode} vertex in $BG$;\label{lin:adde4}
			\If{(\ec.$ts_i$ $\neq$ toNode.$ts_i$)}\label{lin:adde5}
			\State Create new Graph Node (or \egn) in \el{};\label{lin:adde6} 
			\If{(\ep.\en.CAS(\ec, \egn))}\label{lin:adde7}
			\State Increment the \inc{} atomically of \blank{1.1cm} \egn.\emph{vref} in \vl{};\label{lin:adde8}
			\State return$\langle$\emph{Edge added}$\rangle$; 
			\label{lin:adde9}
			\EndIf\label{lin:adde10}
			\State goto \Lineref{adde4}; /*Start with the \ep{} to identify \blank{.7cm} the new $\langle$\ep, \ec{}$\rangle$*/\label{lin:adde11}
			\Else\label{lin:adde12}
			\State return$\langle$\emph{Edge already present}$\rangle$; 
			\label{lin:adde13}
			\EndIf\label{lin:adde14}
			\EndProcedure\label{lin:adde15}
		\end{algorithmic}
	\end{multicols}
\end{algorithm}

\begin{algorithm}[H]
	\scriptsize
	\caption{\emph{\searchl{(cacheVer, $AU_{id}$)}}: Thread searches source node in \cachel{}.}
	\setlength{\multicolsep}{0pt}
	\begin{multicols}{2}	
		\begin{algorithmic}[1]
			\makeatletter\setcounter{ALG@line}{39}\makeatother
			\Procedure{\emph{\searchl{($cacheVer$)}}}{}
			\If{(cacheVer.\inc.CAS(0, -1))} \label{lin:sl1} 
			\State \nc{} $\gets$ \nc{}.$get\&Inc()$; \label{lin:sl2}
			\State $AU_{id}$ $\gets$ cacheVer.$AU_{id}$;
			\State return$\langle$cacheVer$\rangle$;\label{lin:sl5}
			\Else\label{lin:sl6}
			\State return$\langle nil \rangle$;\label{lin:sl7}
			\EndIf\label{lin:sl8}
			\EndProcedure
		\end{algorithmic}
	\end{multicols}
\end{algorithm}

\begin{algorithm}[H]
	\scriptsize
	\caption{\emph{\searchg{(BG, $AU_{id}$)}}: Thread searches the source node in BG.}
	\setlength{\multicolsep}{0pt}
	\begin{multicols}{2}
		\begin{algorithmic}[1]
			\makeatletter\setcounter{ALG@line}{48}\makeatother
			\Procedure{\emph{\searchg{(BG, $AU_{id}$)}}}{}
			\State \vnode{} $\gets$ BG.\vh;	
			\While{(\vnode.\vn{} $\neq$ BG.\vt)} 
			\If{(\vnode.\inc.CAS(0, -1))}\label{lin:sg1}
			\State \nc{} $\gets$ \nc{}.$get\&Inc()$; \label{lin:sg2}
			\State $AU_{id}$ $\gets$ \vnode.$AU_{id}$;
			\State return$\langle \vnode \rangle$;\label{lin:sg5}
			\EndIf\label{lin:sg8}
			\State \vnode $\gets$ \vnode.\vn;	
			\EndWhile
			\State return$\langle nil \rangle$;\label{lin:sg7}
			\EndProcedure
		\end{algorithmic}
	\end{multicols}
\end{algorithm}

\begin{algorithm}[H]
	\scriptsize
	\caption{\emph{decInCount(remNode)}: Decrement the \inc{} of each conflicting node.}
	\setlength{\multicolsep}{0pt}
	\begin{multicols}{2}
		\begin{algorithmic}[1]
			\makeatletter\setcounter{ALG@line}{60}\makeatother
			\Procedure{\emph{decInCount(remNode)}}{}
			\While{(remNode.\en $\neq$ remNode.\et)}	
			\State Decrement the \emph{inCnt} atomically of \blank{.7cm} remNode.\emph{vref} in the \vl{}; \label{lin:ren1}
			\If{(remNode.\emph{vref}.\inc{} == 0)}\label{lin:ren2}
			\State Add remNode.\emph{verf} node into \cachel{} of \blank{1.1cm} thread local log, \tl{};\label{lin:ren3}
			\EndIf\label{lin:ren4}
			\State remNode $\gets$ remNode.\en.\emph{verf};		
			\State return$\langle$remNode$\rangle$;
			\EndWhile
			\State return$\langle nil \rangle$;
			\EndProcedure		
		\end{algorithmic}
	\end{multicols}
\end{algorithm}

\begin{algorithm}[H]
	\scriptsize
	\label{alg:exec} 
	\caption{\exec{\emph{(curAU)}}: Execute the current atomic-units.}
	\setlength{\multicolsep}{0pt}
	\begin{multicols}{2}
		\begin{algorithmic}[1]
			\makeatletter\setcounter{ALG@line}{71}\makeatother
			\Procedure{\exec{($curAU$)}}{}
			\While{(curAU.steps.hasNext())} /*Assume that \blank{.35cm} curAU is a list of steps*/
			\State curStep = currAU.steps.next(); 
			\Switch{(curStep)}
			\EndSwitch
			\Case{read($x$):}
			\State Read data-object $x$ from a shared memory;
			\EndCase
			\Case{write($x, v$):} 
			\State Write $x$ in shared memory with value $v$;
			\EndCase
			\Case{default:}
			\State /*Neither read or write in shared memory*/;
			\State execute curStep;
			\EndCase
			\EndWhile			
			\State return $\langle void \rangle$	
			\EndProcedure			
		\end{algorithmic}
	\end{multicols}
\end{algorithm}

After successfully adding \vnode{} in the BG, concurrent miner calls \texttt{addEdge()} method to add the conflicting node (or \egn{}) corresponding to \vnode{} in the \el{}. First, the \texttt{addEdge()} method identifies the correct location of \egn{} in the \el{} of corresponding \vnode{} at \Lineref{adde4}. If \egn{} is not part of \el{}, it creates and adds it into \el{} of \vnode{} at \Lineref{adde7} in a lock-free manner using atomic CAS operation. After successful addition of \enode{} in the \el{} of \vnode{}, it increments the \inc{} of \enode.$vref$ (to maintain indegree count) node, which is present in the \vl{} at \Lineref{adde8}.

\noindent
\textbf{Lock-free Graph Library Methods Accessed by Concurrent Validator:} Concurrent validator uses \texttt{searchLocal(), searchGlobal()} and \texttt{decInCount()} methods of lock-free graph library. First, concurrent validator thread calls \texttt{searchLocal()} method to identify the source node (having indegree (or \inc) 0) in its local \cachel{} (or thread-local memory). If any source node exists in the local \cachel{} with \inc{} 0, then to claim that node, it sets the \inc{} field to -1 at \Lineref{sl1} atomically. 

If the source node does not exist in the local \cachel{}, then the concurrent validator thread calls \texttt{searchGlobal()} method to identify the source node in the BG at \Lineref{sg1}. If a source node exists in the BG, it sets \inc{} to -1 atomically to claim that node and calls the \texttt{decInCount()} method to decreases the \inc{} of all conflicting nodes atomically, which are present in the \el{} of corresponding source node at \Lineref{ren1}. While decrementing \inc{s}, it checks if any conflicting node became a source node, then it adds that node into its local \cachel{} to optimize the search time of identifying the next source node at \Lineref{ren3}. 


\subsection{Concurrent Miner}
\label{subsec:cminer}
Smart contracts in \bc{} are executed in two different contexts. First, the \Miner{} proposes a new block. Second, multiple \Validator{s} re-execute to verify and validate the block proposed by the \Miner. In this subsection, we describe how miner executes the \SContract{s} concurrently.

\setlength{\intextsep}{0pt}
\begin{algorithm}[!t]
	\scriptsize 
	\caption{\emph{\cminer{(\aul, STM)}}: Concurrently $m$ threads are executing atomic-units from \aul{} (or list of atomic-units) with the help of STM.}
	\label{alg:cminer}	
	\setlength{\multicolsep}{0pt}
	\begin{multicols}{2}
		\begin{algorithmic}[1]
			\makeatletter\setcounter{ALG@line}{85}\makeatother
			\Procedure{\emph{\cminer}{(\aul, STM)}}{}\label{lin:cminer1}
			\State {/*Add all AUs in the Concurrent Bin (\emph{concBin[]})*/
				\State \emph{concBin[]} $\gets$ \aul;} \label{lin:cminer111}
			\State /*curAU is the current AU taken from \aul*/
			\State curAU $\gets$ $curInd$.$get\&Inc(\aul)$; \label{lin:cminer2}
			\State /*Execute until all AUs successfully completed*/\label{lin:cminer3}
			\While{(curAU $<$ size\_of(\aul))}\label{lin:cminer4}
			\State $T_i$ $\gets$ STM.\begtrans{()};
			\label{lin:cminer5}
			\While{(curAU.steps.hasNext())} 
			\label{lin:cminer6}
			\State curStep = currAU.steps.next(); 
			\label{lin:cminer7}
			\Switch{(curStep)}\label{lin:cminer8}\EndSwitch
			\Case{read($x$):}\label{lin:cminer9}
			\State $v$ $\gets$ STM.\readi{($x$)}; 
			\label{lin:cminer10}
			\If{($v$ == $abort$)}\label{lin:cminer11}
			\State goto \Lineref{cminer5};\label{lin:cminer12}
			\EndIf\label{lin:cminer13}
			\EndCase
			\Case{write($x, v$):} \label{lin:cminer14}
			\State STM.$write_i$($x, v$); \label{lin:cminer16}
			\EndCase
			\Case{default:}\label{lin:cminer17}
			\State /*Neither read or write in memory*/\label{lin:cminer18}
			\State execute curStep;\label{lin:cminer19}
			\EndCase
			\EndWhile\label{lin:cminer20}
			\State /*Try to commit the current transaction $T_i$ and \blank{.7cm} update the \cl{[i]}*/
			\label{lin:cminer21}
			\State $v$ $\gets$ \tryc{$_i$()}; \label{lin:cminer22}
			\If{($v == abort$)}\label{lin:cminer23}
			\State goto \Lineref{cminer5};\label{lin:cminer24}
			\EndIf			\label{lin:cminer25}
			
			\If{($\cl[i] == nil$)}\label{lin:cminer251}
			\State {curAU doesn't have dependencies with other \blank{1cm} AUs. So, no need to create a node in BG.}
			\Else
			\State {create a nodes with respective dependencies \blank{1cm} from curAU to all AUs $\in$ $\cl[i]$ in BG \blank{1cm} and remove curAU and AUs from \emph{concBin[]}} \label{lin:cminer2511}
			\State Create \vnode{} with $\langle$\emph{$i$, $AU_{id}$, 0, nil, nil}$\rangle$ as \blank{1cm} a vertex of Block Graph;	\label{lin:cminer26}
			\State BG(\emph{vNode}, STM);		\label{lin:cminer27}
			\EndIf
			\State curAU $\gets$ $curInd$.$get\&Inc(\aul)$; \label{lin:cminer28}
			\EndWhile	\label{lin:cminer29}
			\EndProcedure\label{lin:cminer30}
			
		\end{algorithmic}
	\end{multicols}
\end{algorithm}

\cmnt{\tikz \node[circle,scale=.5,color=black, fill=pink]{\textbf{1}};}A \emph{concurrent miner} gets the set of transactions from the blockchain network. Each transaction is associated with a method (\au{}) of smart contracts. To run the \SContract{s} concurrently, we have faced the challenge of identifying the conflicting transactions at run-time because \SContract{} languages are Turing-complete. Two transactionsconflict if they access a shared data-objects and at least one of them perform write operation. \cmnt{\tikz \node[circle,scale=.5,color=black, fill=pink]{\textbf{2}};}In \conminer{}, conflicts are identified at run-time using an efficient framework provided by the optimistic software transactional memory system (STMs). STMs access the shared data-objects called as \tobj{s}. Each shared \tobj{} is initialized to an initial state (or IS). The \au{s} may modify the IS to some other valid state. Eventually, it reaches the final state (or FS) at the end of block-creation. As shown in \algoref{cminer}, the concurrent miner first copies all the AUs in the concurrent bin at \Lineref{cminer111}. Each transaction $T_i$ gets the unique timestamp $i$ from \texttt{STM.begin()} at \Lineref{cminer5}. Then transaction $T_i$ executes the \au{} of \SContract{s}. \emph{Atomic-unit} consists of multiple steps such as $reads$ and $writes$ on shared \tobj{s} as $x$. Internally, these $read$ and $write$ steps are handled by the \texttt{STM.read()} and \texttt{STM.write()}, respectively. At \Lineref{cminer9}, if current \au{} step (or curStep) is $read(x)$ then it calls the \texttt{STM.read(x)}. Internally, \texttt{STM.read()} identify the shared \tobj{} $x$ from transactional memory (or TM) and validate it. If validation is successful, it gets the value as $v$ at \Lineref{cminer10} and executes the next step of \au{}; otherwise, re-executed the \au{} if $aborted$ at \Lineref{cminer11}. 

\cmnt{
\begin{figure}
	\scalebox{.32}{\input{figs/cminer.pdf_t}}
	\centering
	\caption{Execution of Concurrent Miner}
	\label{fig:cminer}
\end{figure}
}

\cmnt{
\begin{figure*}
	\centering
	\centerline{\scalebox{0.45}{\input{figs/cminer.pdf_t}}}
	\caption{Execution of Concurrent Miner}
	\label{fig:cminer}
\end{figure*}
}

If curStep is $write(x)$ at \Lineref{cminer14} then it calls the \texttt{STM.write(x)}. Internally, \texttt{STM.write()} stores the information of shared \tobj{} $x$ into local log (or \txlog) in write-set (or $wset_i$) for transaction $T_i$. We use an optimistic approach in which the transaction's effect will reflect onto the TM after the successful \texttt{STM.tryC()}. If validation is successful for all the $wset_i$ of transaction $T_i$ in \texttt{STM.tryC()}, i.e., all the changes made by the $T_i$ are consistent, then it updates the TM; otherwise, re-execute the \au{} if $aborted$ at \Lineref{cminer23}. After successful validation of \texttt{STM.tryC()}, it also maintains the conflicting transaction of $T_i$ into the conflict list in TM.

\cmnt{
\begin{figure*}
	\centerline{\scalebox{0.45}{\input{figs/graphs.pdf_t}}}
	\caption{Pictorial representation of \confg{}, CG}
	\label{fig:confg}
\end{figure*}
}
\cmnt{
Once the transaction commits it stores its conflicts in the form of \CG{}. In order to store the \emph{Conflict Graph} or \confg \emph{(CG(V, E))}, we are maintaining \emph{adjacency list}. In which all the vertices (or \vnode{s}) are stored in the vertex list (or \vl{}) in increasing order of timestamp between the two sentinal node \vh{} (-$\infty$) and \vt{} (+$\infty$). The structure of the vertex node (or \vnode) is as $\langle ts_i, AU_{id}, \inc{}, \vn, \en\rangle$. Where $ts$ is a unique timestamp $i$ of committed transactions $T_i$ assign at the beginning of it. $AU_{id}$ is the $id$ of \au{} which is executed by transaction $T_i$. To maintain the indegree count of each \vnode{} we use \inc{} initialize as 0. \vn{} points to the next vertex of the \vl. For corresponding to each \vnode{}, it maintains all the conflicts in the edge list (\el) as shown in \figref{confg}. \el{} stores \enode{} (or conflicting transaction nodes) belonging to \emph{vNode} stores in increasing order of timestamp (or $ts$) in between the two sentinal nodes \eh{} (-$\infty$) and \et{} (+$\infty$). The structure of \enode{} is as $\langle$\emph{$ts_j$, vref}, \en{} $\rangle$. $ts_j$ is a unique timestamp $j$ of \emph{committed} transaction $T_j$ which had a conflict with $T_i$. \emph{Vertex node reference (or vref)} keeps the reference of its own vertex which is present in the \vl. \en{} points the next \enode{} of the \el{}.
}


\cmnt{\tikz \node[circle,scale=.5,color=black, fill=pink]{\textbf{3}};} 
If the conflict list is \emph{nil} (\Lineref{cminer251}), there is no need to create a node in the BG. Otherwise, create the node with respective dependencies in the BG and remove those AUs from the concurrent bin (\Lineref{cminer2511}). To maintain the BG, it calls \texttt{addVert()} and \texttt{addEdge()} methods of the lock-free graph library. The details of \texttt{addVert()} and \texttt{addEdge()} methods are explained in \subsecref{bg}.
Once the transactions successfully executed the \au{s} and done with BG construction, the \conminer{} computes the hash of the previous block. Eventually, \cmnt{\tikz \node[circle,scale=.5,color=black, fill=pink]{\textbf{5}};}\conminer{} proposes a block consisting of a set of transactions, BG, the final state of each shared \tobj{s}, previous block hash, and\cmnt{\tikz \node[circle,scale=.5,color=black, fill=pink]{\textbf{6}};} sends it to all other network peers to validate. 


\subsection{Concurrent Validator}
\label{subsec:cvalidator}
The concurrent validator validates the block proposed by the concurrent miner. 
It executes the block transactions concurrently and deterministically in two phases using a concurrent bin and BG given by the \conminer{}. In the first phase, validator threads execute the independent AUs of concurrent bin concurrently (\Lineref{val11} to \Lineref{val112}). Then in the second phase, it uses BG to executes the dependent AUs by \texttt{executeCode()} method at \Lineref{sl4} and \Lineref{sl41} using \texttt{searchLocal()}, \texttt{searchGlobal()} and \texttt{decInCount()} methods of lock-free graph library at \Lineref{val5}, \Lineref{val16} and (\Lineref{val8}, \Lineref{val20}), respectively. BG consists of dependency among the conflicting transactions that restrict them to execute serially.
The functionality of lock-free graph library methods is explained earlier in \subsecref{bg}. 



\begin{algorithm}[!t]
	\scriptsize
	\caption{\emph{\cvalidator}{(\aul, \emph{BG)}}: Concurrently $V$ threads are executing AUs with the help of concurrent bin followed by the BG given by the miner.}
	\setlength{\multicolsep}{0pt}
	\begin{multicols}{2}
		\begin{algorithmic}[1]
			\makeatletter\setcounter{ALG@line}{122}\makeatother
			\Procedure{\emph{\cvalidator}{(\aul, BG)}}{}		
			\State /*Execute until all AUs successfully completed*/ \label{lin:val1}
			
			{
				\State /*\textbf{Phase-1}: Concurrent Bin AUs execution.*/
				\While{(concCount $<$ size\_of(\emph{concBin[]}))}\label{lin:val11}
				\State count $\gets$ concCount.$get\&Inc(\aul)$;
				\State $AU_{id}$ $\gets$ \emph{concBin[count]};
				\State  \exec{($AU_{id}$)}; 
				\EndWhile \label{lin:val112}
				\State /*\textbf{Phase-2}: Block Graph AUs execution.*/
			}
			\While{(\nc{} $<$ size\_of(\aul))} \label{lin:val2}
			\While{(\cachel{}.hasNext())} 
			\label{lin:val3}
			\State cacheVer $\gets$ \cachel{}.next(); \label{lin:val4}
			\State cacheVertex $\gets$  \searchl{(cacheVer, \blank{1cm} AU$_{id}$)};\label{lin:val5}
			\State  \exec{($AU_{id}$)};\label{lin:sl4} 
			\While{(cacheVertex)}\label{lin:val7}
			\State cacheVertex $\gets$ decInCount(cacheVertex);\label{lin:val8}
			\EndWhile\label{lin:val10}
			\State Remove the current node (or cacheVertex) \blank{1cm} from local \cachel; \label{lin:val12}
			\EndWhile\label{lin:val13}
			\State vexNode $\gets$ \searchg{(BG, $AU_{id}$)}; 
			\label{lin:val16}
			\State  \exec{($AU_{id}$)};\label{lin:sl41} 
			\While{(verNode)}\label{lin:val19}
			\State verNode $\gets$ decInCount(verNode);\label{lin:val20}
			\EndWhile\label{lin:val22}
			\EndWhile\label{lin:val27}
			\EndProcedure
		\end{algorithmic}
	\end{multicols}
	
\end{algorithm}

After the successful execution of all the \au{s}, the \convalidator{} compares its computed final state with the final states given by the \conminer{}. If the final state matches for all the shared data-objects, then the block proposed by the \conminer{} is valid. Finally, {based on consensus between network peers, the block is appended to the blockchain, and the respective \conminer{} is rewarded.}

\subsection{Optimizations}
\label{subsec:opt}
{
To make the proposed approach storage optimal and efficient, this subsection explains the key change performed on top of the solution proposed by Anjana et al.~\cite{Anjana+:CESC:PDP:2019}.


In Anjana et al.~\cite{Anjana+:CESC:PDP:2019}, there is a corresponding vertex node in the \bg{} (BG) for every \sctrn{s} in the block. We observed that all the \sctrn{s} in the block need not have dependencies. Adding a vertex node for such \sctrn{s} takes additional space in the block. This is the first optimization our approach provides. In our approach, only the dependent \sctrn{s} have a vertex in the BG, while the independent \sctrn{s} are stored in the concurrent bin, which does not need any additional space. During the execution, a concurrent miner thread does not add a vertex to the BG if it identifies that the currently executed \sctrn{} does not depend on the \sctrn{s} already executed. However, suppose any other miner thread detects any dependence during the remaining \sctrn{s} execution. That thread will add the dependent \sctrn{s} vertices in the BG.


For example, let say we have $n$ \sctrn{s} in a block and a vertex node size is $\approx m$ kb to store in the BG, then it needs a total of $n*m$ kb of vertex node space for Anjana et al.~\cite{Anjana+:CESC:PDP:2019}. Suppose from $n$ \sctrn{s}, only $\frac{n}{2}$ have the dependencies, then a total of $\frac{n}{2}*m$ kb vertex space needed in the BG. In the proposed approach, the space optimization can be 100\% in the best case when all the \sctrn{s} are independent. While in the worst case, it can be 0\% when all the \sctrn{s} are dependent. However, only a few \sctrn{s} in a block have dependencies. Space-optimized BG helps to improve the network bandwidth and reduces network congestion.

Further, our approach combines the benefit of both \specbin{}-based approach \cite{VikramHerlihy:EmpSdy-Con:Tokenomics:2019} and STM-based approach~\cite{Anjana+:CESC:PDP:2019} to yield maximum speedup that can be achieved by validators to execute \sctrn{s}. So, another optimization is at the validators side; due to the concurrent bin in the block, the time taken to traverse the BG will decrease; hence, speedup increases. The concurrent validators execution is modified and divided into two phases. First, it concurrently executes \sctrn{s} of the concurrent bin using multiple threads, since \sctrn{s} in the concurrent bin will be independent. While in the second phase, dependent \sctrn{s} are stored in the BG and concurrently executed using BG to preserve the transaction execution order as executed by the miner.

}

\section{Correctness}
\label{sec:correctness}
{The correctness of concurrent BG, miner, and validator is described in this section. We first list the linearization points (LPs) of the \bg{} library methods as follows:}

\begin{enumerate}
	\item \addv{(\vgn)}: (\vp.\vn.CAS(\vc, \vgn)) in \Lineref{addv5} is the LP point of \addv{()} method if \vnode{} is not exist in the BG. If \vnode{} is exist in the BG then (\vc.$ts_i$ $\neq$ \vgn.$ts_i$) in \Lineref{addv3} is the LP point.
	\item \adde{\emph{(fromNode, toNode)}}: 		(\ep.\en.CAS(\ec, \egn)) in \Lineref{adde7} is the LP point of \adde{()} method if \enode{} is not exist in the BG. If \enode{} is exist in the BG then 		(\ec.$ts_i$ $\neq$ toNode.$ts_i$) in \Lineref{adde5} is the LP point.
	\item \searchl{(cacheVer, $AU_{id}$)}: (cacheVer.\inc.CAS(0, -1)) in \Lineref{sl1} is the LP point of \searchl{()} method.
	\item \searchg{(BG, $AU_{id}$)}: (\vnode.\inc.CAS(0, -1)) in \Lineref{sg1} is the LP point of \searchg{()} method.
	\item \texttt{decInCount(remNode)}: \Lineref{ren1} is the LP point of \texttt{decInCount()} method.
\end{enumerate}

\begin{theorem}
	Any history $H_m$ generated by the concurrent miner using the BTO protocol satisfies co-opacity.
\end{theorem}
\begin{proof}
	Concurrent miner executes \sctrn{s} concurrently using BTO protocol and generate a concurrent history $H_m$. The underlying BTO protocol ensures the correctness of concurrent execution of $H_m$. The BTO protocol \cite[Chap 4]{WeiVoss:TIS:2002:Morg} proves that any history generated by it satisfies co-opacity \cite{Peri+:OSTM:Netys:2018}. So, implicitly BTO proves that the history $H_m$ generated by concurrent miner using BTO satisfies co-opacity. 
\end{proof}

\begin{theorem}
	Any history $H_m$ generated by the concurrent miner using the MVTO protocol satisfies opacity.
\end{theorem}
\begin{proof}
	Concurrent miner executes \sctrn{s} concurrently using MVTO protocol and generate a concurrent history $H_m$. The underlying MVTO protocol ensures the correctness of concurrent execution of $H_m$. The MVTO protocol \cite{Kumar+:MVTO:ICDCN:2014} proves that any history generated by it satisfies opacity \cite{GuerKap:Opacity:PPoPP:2008}. So, implicitly MVTO proves that the history $H_m$ generated by concurrent miner using MVTO satisfies opacity.
\end{proof}

\begin{theorem}
	All the dependencies between the conflicting nodes are captured in BG.
\end{theorem} 
\begin{proof}
	Dependencies between the conflicting nodes are captured in the BG using LP points of lock-free graph library methods defined above. Concurrent miner constructs the lock-free BG using BTO and MVTO protocol in \subsecref{bg}. BG consists of vertices and edges, where each committed \sctrn{} act as a vertex  and edges (or dependencies) represents the conflicts of the respective STM protocol (BTO and MVTO). As we know, STM protocols BTO \cite[Chap 4]{WeiVoss:TIS:2002:Morg} and MVTO \cite{Kumar+:MVTO:ICDCN:2014} used in this paper for the concurrent execution are correct, i.e., these protocols captures all the dependencies correctly between the conflicting nodes. Hence, all the dependencies between the conflicting nodes are captured in the BG.
\end{proof}

\begin{theorem}
	\label{thm:hmve}
	A history $H_m$ generated by the concurrent miner using BTO protocol and a history $H_v$ generated by a concurrent validator are view equivalent.
\end{theorem}
\begin{proof}
	A concurrent miner executes the \sctrn{s} of $H_m$ concurrently using BTO protocol, captures the dependencies of $H_m$ in the BG, and proposes a block $B$. Then it broadcasts the block $B$ along with BG to concurrent validators to verify the block $B$. The concurrent validator applies the topological sort on the BG and obtained an equivalent serial schedule $H_v$. Since the BG constructed from $H_m$ considers all the conflicts and $H_v$ obtained from the topological sort on the BG. So, $H_v$ is equivalent to $H_m$. Similarly, $H_v$ also follows the \emph{read from} relation of $H_m$. Hence, $H_v$ is legal. Since $H_v$ and $H_m$ are equivalent to each other, and $H_v$ is legal. So, $H_m$ and $H_v$ are view equivalent. 
\end{proof}

\begin{theorem}
	A history $H_m$ generated by the concurrent miner using MVTO protocol and a history $H_v$ generated by a concurrent validator are multi-version view equivalent.
\end{theorem}
\begin{proof}
	Similar to the proof of \thmref{hmve}, the concurrent miner executes the \sctrn{s} of $H_m$ concurrently using MVTO protocol, captures the dependencies in the BG, proposes a block $B$, and broadcasts it to the concurrent validators to verify it. MVTO maintains multiple-version corresponding to each shared object. Later, concurrent validator obtained $H_v$ by applying topological sort on the BG provided by the concurrent miner. Since, $H_v$ obtained from topological sort on the BG so, $H_v$ is equivalent to $H_m$. Similarly, the BG maintains the \emph{read from} relations of $H_m$. So, from MVTO protocol if $T_j$ reads a value for shared object $k$ say $r_j(k)$  from $T_i$ in $H_m$ then $T_i$ committed before $r_j(k)$ in $H_v$. Therefore, $H_v$ is valid. Since $H_v$ and $H_m$ are equivalent to each other and $H_v$ is valid. So, $H_m$ and $H_v$ are multi-version view equivalent. 
\end{proof}

\section{Experimental Evaluation}
\label{sec:opt-result}


We aim to increase the efficiency of the miners and validators by employing concurrent execution of \sctrn{s} while optimizing the size of the BG appended by the miner in the block. To assess the efficiency of the proposed approach, we performed simulation on the series of benchmark experiments with Ethereum~\cite{ethereum:url} smart contracts from Solidity documentation~\cite{Solidity}. Since multi-threading is not supported by the Ethereum Virtual Machine (EVM)~\cite{ethereum:url, Dickerson+:ACSC:PODC:2017}, we converted the Ethereum smart contracts into C++. We evaluated the proposed approach with the state-of-the-art approaches \cite{Anjana+:CESC:PDP:2019, Dickerson+:ACSC:PODC:2017, VikramHerlihy:EmpSdy-Con:Tokenomics:2019} over baseline serial execution on three different workloads by varying the number of \sctrn{s}, the number of threads, and the number of shared objects. {The benchmark experiments are conservative and consist of one or fewer smart contracts \sctrn{s} in a block, }
which leads to a higher degree of conflicts than actual conflicts in practice where a block consists of \sctrn{s} from different contracts ($\approx$ 1.5 million deployed smart contracts \cite{EthereumByNumbers}). Due to fewer conflicts in the actual blockchain, the proposed approach is expected to provide greater concurrency. We structure our experimental evaluation to answer the following questions: 
\begin{enumerate}
	\item How much speedup is achieved with varying \sctrn{s} by concurrent miners and validators when fixing the number of threads and shared objects? As conflicts increase with increasing \sctrn{s}, we expect a decrease in speedup.
	
	\item How does speedup change when increasing the number of threads with a fixed number of \sctrn{s} and shared objects? We expect to see the speedup increase with increasing threads confined by logical threads available within the system. 
	\item How does speedup shift over different shared objects with fixed \sctrn{s} and threads? We expect an increase in speedup due to conflict deterioration with objects increase. So, we anticipate concurrent miners and validators overweigh serial miners and validators with fewer conflicts. 
\end{enumerate}

\subsection{Contract Selection and Benchmarking}

{This section provides a comprehensive overview of benchmark contracts coin, ballot, and simple auction from Solidity Documentation~\cite{Solidity} selected as real-world examples for evaluating the proposed approach.}
\cmnt{The selected contacts are converted to C++ for concurrent execution. We chose different contracts to demonstrate the wide variety of use-cases, such as a financial application using a coin contract, a collaborative use-case using a ballot, and a bidding application using an auction contract.}The \sctrn{s} in a block for the coin, ballot, and auction benchmark operate on the same contract, i.e., consists of the transaction calls of one or more methods of the same contract. In practice, a block consists of the \sctrn{s} from different contracts; hence we designed another benchmark contract called \emph{mix contract} consisting of contract transactions from coin, ballot, and auction in equal proportion in a block. 
The benchmark contracts and respective methods are as follows:

\noindent
\textbf{Coin Contract:} The coin contract is the simplest form of sub-currency. \cmnt{It implements three distinct functions identical to the standard ERC-20 token contract standard.}The users involved in the contract have accounts, and \emph{accounts} are shared objects. It implements methods such as \texttt{mint()}, \texttt{transfer()/send()}, and \texttt{getbalance()} which represent the \sctrn{s} in a block. The contract deployer uses the \texttt{mint()} method to give initial coins/balance to each account with the same fixed amount. We initialized the coin contract's initial state with a fixed number of accounts on all benchmarks and workloads. Using \texttt{transfer()}, users can transfer coin from one account to other account. The \texttt{getbalance()} is used to check the coins in a user account. For the experiments a block consists of 75\% \texttt{getbalance()}, and 25\% \texttt{transfer()} calls. A conflict between \sctrn{s} occurs if they access a common object (account), and at least one of them performs a \texttt{transfer()} operation.



\noindent
\textbf{Ballot Contract:} The ballot contract is an electronic voting contract in which \emph{voters} and \emph{proposals} are shared objects. The \texttt{vote()}, \texttt{delegate()}, and \texttt{winningproposal()} are the methods of ballot contract. The voters use the \texttt{vote()} method to cast their vote to a specific proposal. Alternatively, a voter can delegate their vote to other voter using \texttt{delegate()} method. A voter can cast or delegate their vote only once. At the end of the ballot, the \texttt{winningproposal()} is used to compute the winner. We initialized the ballot contract's initial state with a fixed number of proposals and voters for benchmarking on different workloads for experiments. The proposal to voter ratio is fixed to 5\% to 95\% of the total shared objects. A block consists of 90\% \texttt{vote()}, and a 10\% \texttt{delegate()} method calls followed by a \texttt{winningproposal()} call for the experiments. The \sctrn{s} will conflict if they operate on the same object. So, if two voters \texttt{vote()} for the same proposal simultaneously, then they will conflict.

\noindent
\textbf{Simple Auction Contract:}
It is an online auction contract in which bidders bid for a commodity online. In the end, the amount from the maximum bidder is granted to the owner of the commodity. The \emph{bidders}, \emph{maximum bid}, and \emph{maximum bidder} are the shared object. In our experiments, the initial contract state is a fixed number of bidders with a fixed initial account balance and a fixed period of the auction to end. In the beginning, the maximum bidder and bid are set to null (the base price and the owner can be set accordingly). The bidder uses the contract method \texttt{bid()} to bid for the commodity with their bid amount—the max bid amount and the bidder changes when a bid is higher than the current maximum. A bidder uses the \texttt{withdraw()} method to move the balance of their previous bid into their account. The bidder uses \texttt{bidEnd()} method to know if the auction is over. Finally, when the auction is ended, the maximum bidder (winner) amount is transferred to the commodity owner, and commodity ownership is transferred to the max bidder. For benchmarking in our experiments a block consist of 8\% \texttt{bid()}, 90\% \texttt{withdraw()}, and 2\% \texttt{bidEnd()} method calls. The max bidder and max bid are the conflict points whenever a new bid with the current highest amount occurs.

\noindent
\textbf{Mix Contract:} In this contract, we combine the \sctrn{s} in equal proportion from the above three contracts (coin, ballot, and auction). Therefore, our experiment block consists of an equal number of corresponding contract transactions with the same initial state as initialized in the above contracts.

\subsection{Experimental Setup and Workloads} 
We ran our experiments on a large-scale 2-socket Intel(R) Xeon(R) CPU E5-2690 V4 @ 2.60 GHz with a total of 56 hyper-threads (14 cores per socket and two threads per core) with 32 GB of RAM running Ubuntu 18.04.

In our experiments, we have noticed that speedup varies from contract to contract on different workloads. The speedup on various contracts is not for comparison between contracts. Instead, it demonstrates the proposed approach efficiency on several use-cases in the blockchain. We have considered the following three workloads for performance evaluation:

\begin{enumerate}
	\item In workload 1 (W1), a block consists of \sctrn{s} varies from 50 to 400, fixed 50 threads, and shared objects of 2K. The \sctrn{s} per block in Ethereum blockchain is on an average of 100, while the actual could be more than 200~\cite{Dickerson+:ACSC:PODC:2017}, however a theoretical maximum of $\approx400$~\cite{EthereumGasLimit2020} after a recent increase in the gas limit. Over time, the number of \sctrn{s} per block is increasing. In practice, one block can have less \sctrn{s} than the theoretical cap, which depends on the gas limit of the block and the gas price of the transactions. We will see that in a block, the percentage of data conflicts increase with increasing \sctrn{s}. The conflict within a block is described by different \sctrn{s} accessing a common shared object, and at least one of them performs an update. We have found that the data conflict varies from contract to contract and has a varied effect on speedup.
	\item In workload 2 (W2), we varied the number of threads from 10 to 60 while fixed the \sctrn{s} to 300 and shared objects to 2K. Our experiment system consists of a maximum of 56 hardware threads, so we experimented with a maximum of 60 threads. We observed that the speedup of the proposed approach increases with an increasing number of threads limited by logical threads. 
	\item The number of \sctrn{s} and threads in workload 3 (W3) are 50 and 300, respectively, although the shared objects range from 1K to 6K. This workload is used with each contract to measure the impact of the number of participants involved. Data conflicts are expected to decrease with an increasing number of shared objects; however, the search time may increases. The speedup depends on the execution of the contract; but, it increases with an increasing number of shared objects.
	
\end{enumerate}

\subsection{Analysis}
In our experiments, blocks of \sctrn{s} were generated for each benchmark contract on three workloads: W1 (varying \sctrn{s}), W2 (varying threads), and W3 (varying shared objects). Then, concurrent miners and validators execute the blocks concurrently. The corresponding blocks serial execution is considered as a baseline to compute the speedup of proposed concurrent miners and validators. The running time is collected for 15 iterations (times) with 10 blocks per iteration, and 10 validators validate each block. The first block of each iteration is left as a warm-up run, and a total of 150 blocks are created for each reading. So, each block execution time is averaged by 9. Further, the total time taken by all iterations is averaged by the number of iteration for each reading; the \eqnref{time} is used to compute a reading time.
\begin{equation}
	\alpha_t = \frac{\displaystyle\sum_{i=1}^{n}\displaystyle\sum_{b=1}^{m-1} \beta_t }{n*(m-1)}    
	\label{eq:time}
\end{equation}
Where $\alpha_t$ is an average time for a reading, $n$ is the number of iterations, $m$ is the number of blocks, and $\beta_t$ is block execution time. 


In all plots, figure (a), (b), and (c) correspond to workload W1, W2, and W3, respectively. \figref{miner-speedup-coin} to \figref{miner-speedup-mix} show the speedup achieved by proposed and state-of-the-art concurrent miners over serial miners for all benchmarks and workloads. \figref{decval-speedup-coin} to \figref{decval-speedup-mix} show the speedup achieved by proposed and state-of-the-art concurrent decentralized validators over serial validators for all benchmarks and workloads. \figref{fj-speedup-coin} to \figref{fj-speedup-mix} show speedup achieved by proposed and state-of-the-art concurrent fork-join validators over serial validators. \figref{bgCoin} to \figref{bgMix} show the average number of edges (dependencies) and vertices (\sctrn{s}) in the block graph for respective contracts on all workloads. While \figref{incSizeCoin} to \figref{incSizeMix} show the percentage of additional space required to store the block graph in Ethereum block. A similar observation has been found \cite{DBLP:journals/corr/abs-1809-01326} for the fork-join validator, the average number of dependencies, and space requirement on other contracts.

We observed that speedup for all benchmark contracts follows the roughly same pattern. In the read-intensive benchmarks (coin and mix), speedup likely to increase on all the workloads, while in the write-intensive benchmark (ballot and auction), speedup drop downs on high contention. We also observed that there might not be much speedup for concurrent miners with fewer \sctrn{s} (less than 100) in the block, conceivably due to multi-threading overhead. However, the speedup for concurrent validators generally increases across all the benchmarks and workloads. Fork-join concurrent validators on W2 is an exception in which speedup drops down with an increase in the number of threads since fork-join follows a master-slave approach where a master thread becomes a performance bottleneck. We also observed that the concurrent validators achieve a higher speedup than the concurrent miners. Because the concurrent miner executes the \sctrn{s} non-deterministically, finds conflicting \sctrn{s}, creates concurrent bin and an efficient BG for the validators to execute the \sctrn{s} deterministically.

Our experiment results also show the BG statics and additional space required to store BG in a block of Ethereum blockchain, which shows the space overhead. We compare our proposed approach with the existing speculative bin (Spec Bin) based approach~\cite{VikramHerlihy:EmpSdy-Con:Tokenomics:2019}, the fork-join approach (FJ-Validator)~\cite{Dickerson+:ACSC:PODC:2017} and the approach proposed in~\cite{Anjana+:CESC:PDP:2019} (we call it default/Def approach). The proposed approach combines the benefit of both bin-based and the STM approaches to get maximum benefit for concurrent miners and validators. The proposed approach\footnote{In the figures, legend items in bold.} produces an optimal BG, reduces the space overhead, and outperforms the state-of-the-art approaches.

\figref{miner-speedup-coin}(a) to \figref{miner-speedup-mix}(a) show the speedup for concurrent miner on W1. 
As shown in \figref{miner-speedup-coin}(a) and \figref{miner-speedup-mix}(a) for read-intensive contracts such as in coin and mix contract, the speedup increases with an increase in \sctrn{s}, respectively. While in write-intensive contracts such as ballot and auction contract the speedup does not increase with an increase in \sctrn{s}; instead, it may drop down if \sctrn{s} increases, as shown in \figref{miner-speedup-ballot}(a) and \figref{miner-speedup-auction}(a), respectively. This is because contention increases with an increase in \sctrn{s}.

\begin{figure}[H]
	\centering
	
	\includegraphics[width=1\textwidth]{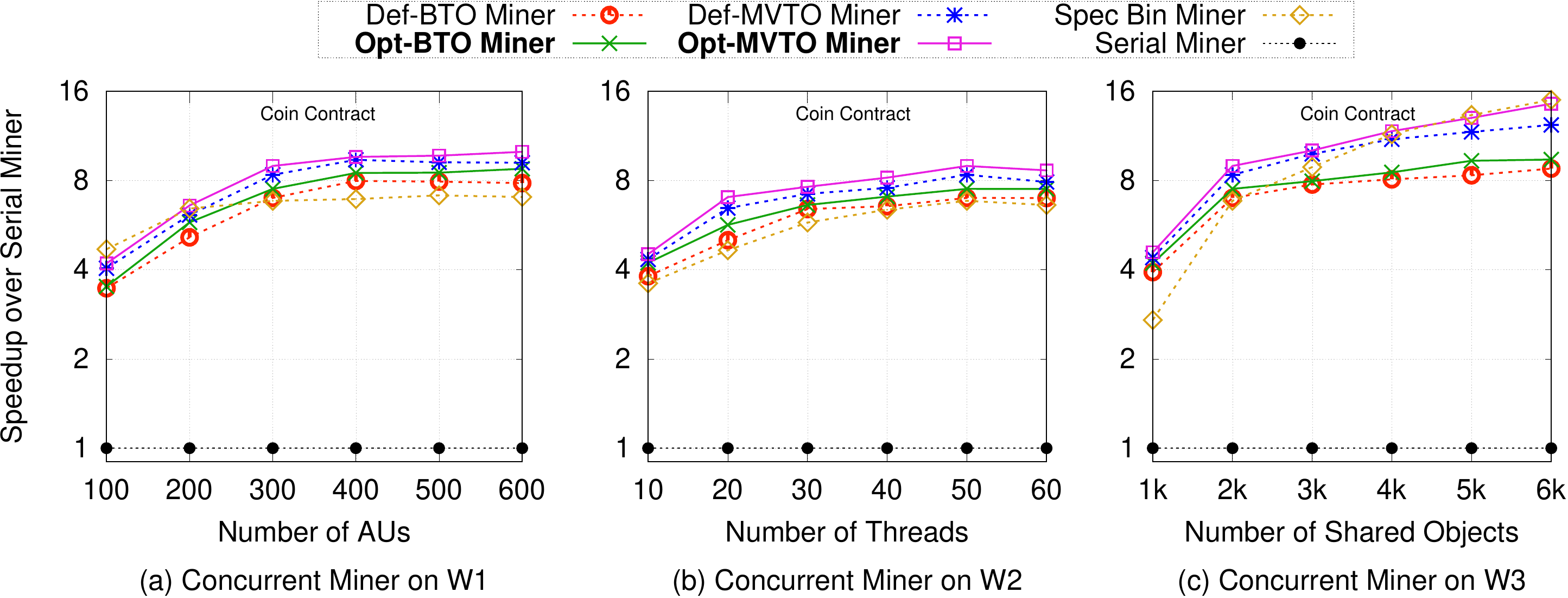}\vspace{-.4cm}
	\caption{Concurrent miner speedup over serial miner for coin contract.}
	\label{fig:miner-speedup-coin}
	
	
	\includegraphics[width=1\textwidth]{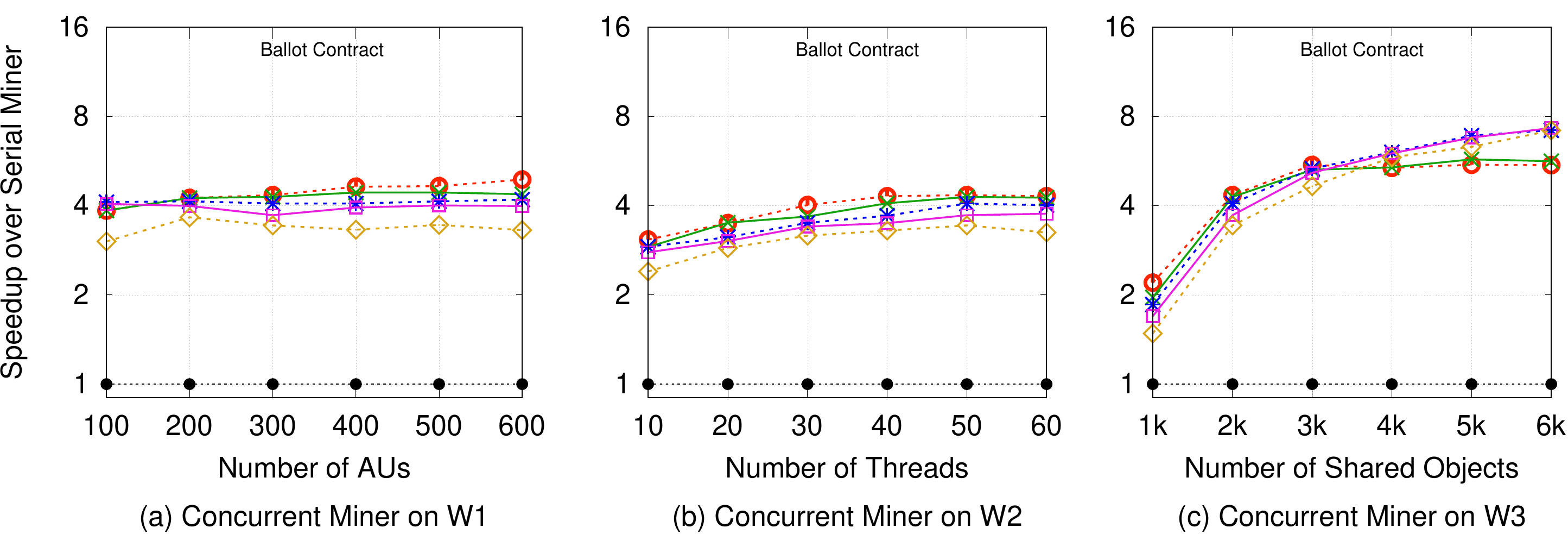}\vspace{-.4cm}
	\caption{ Concurrent miner speedup over serial miner for ballot contract.}
	\label{fig:miner-speedup-ballot}
	
	
	\includegraphics[width=1\textwidth]{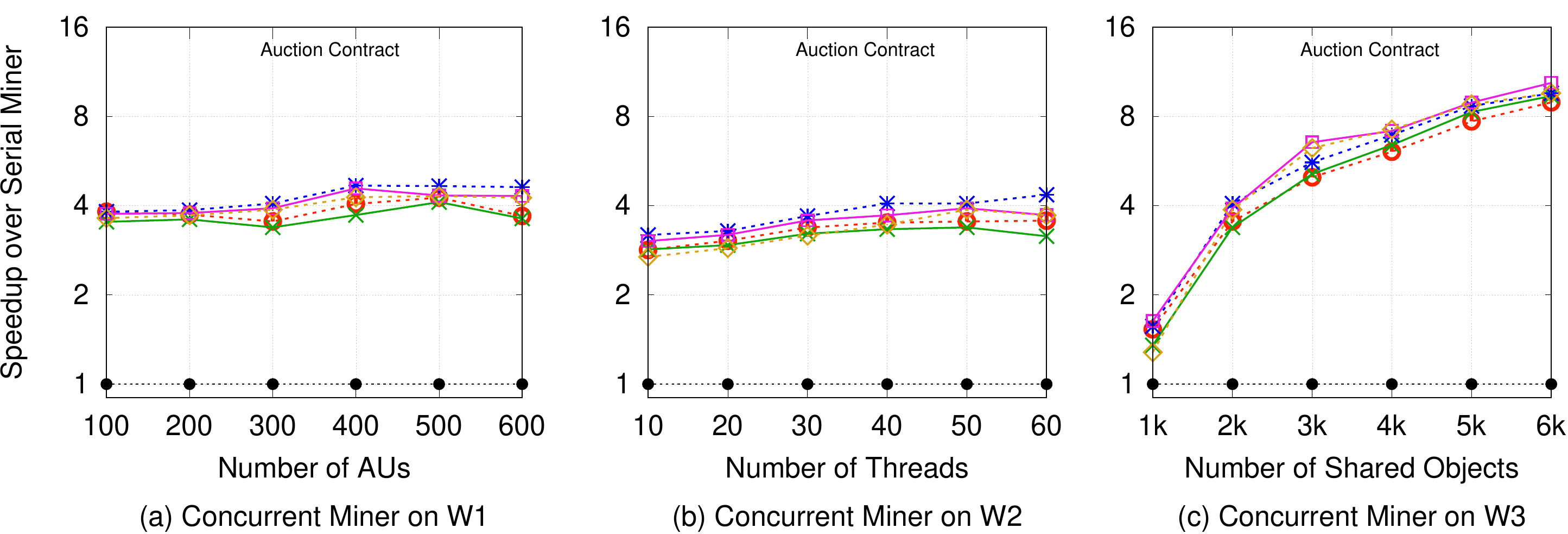}\vspace{-.4cm}
	\caption{Concurrent miner speedup over serial miner for auction contract.}
	\label{fig:miner-speedup-auction}
	
	
	\includegraphics[width=1\textwidth]{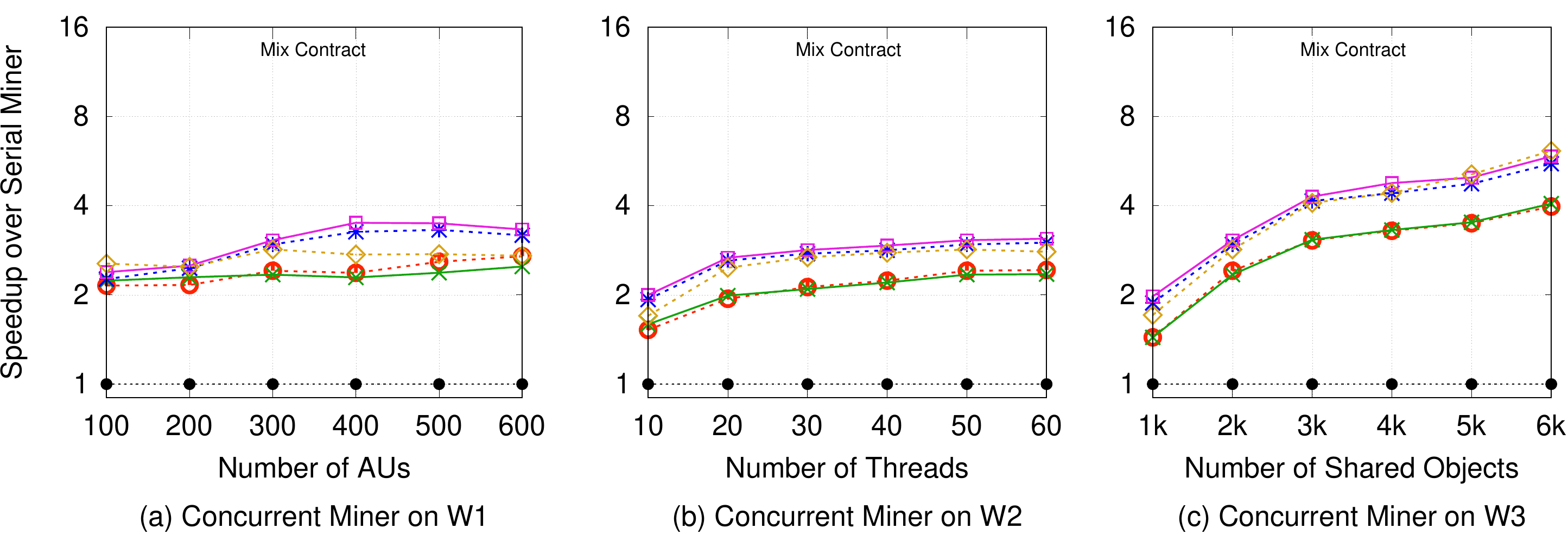}\vspace{-.4cm}
	\caption{Concurrent miner speedup over serial miner for mix contract.}
	\label{fig:miner-speedup-mix}
\end{figure}

\begin{figure}[H]
	\centering
	\includegraphics[width=1\textwidth]{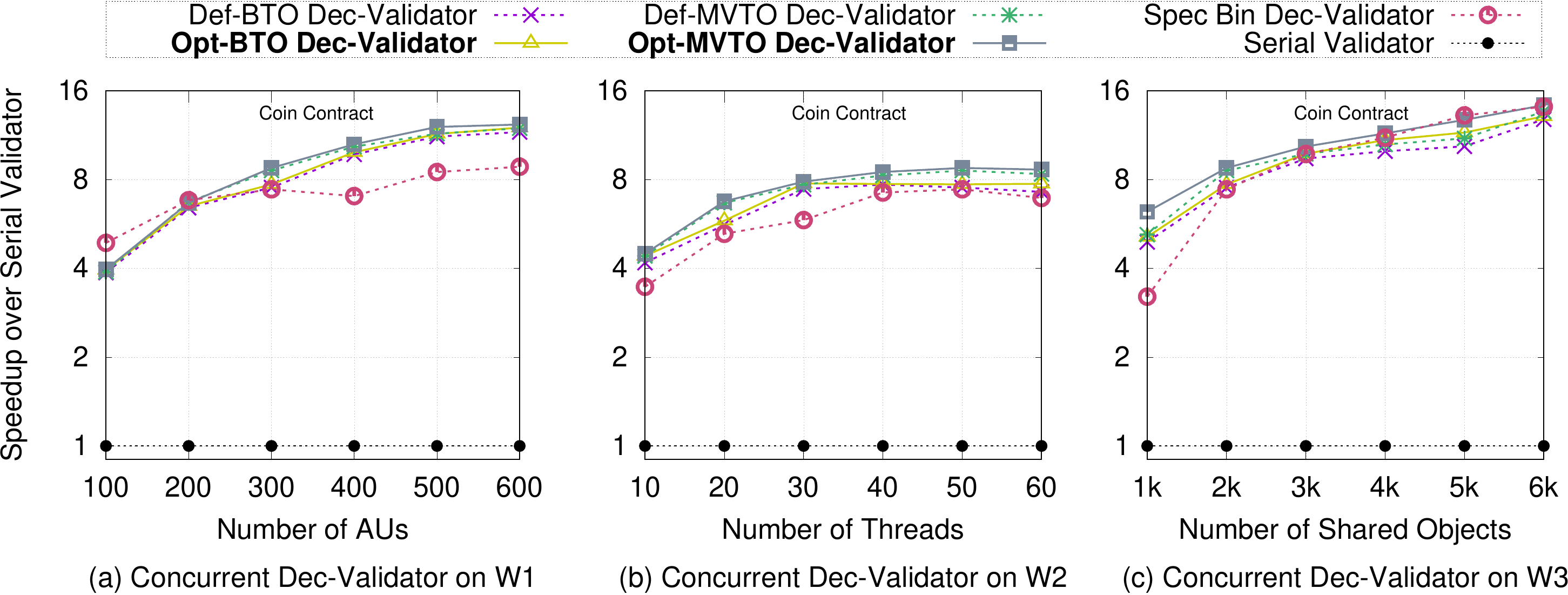}\vspace{-.4cm}
	\caption{Concurrent decentralized validator speedup over serial validator for coin contract.}
	\label{fig:decval-speedup-coin}
	
	
	\includegraphics[width=1\textwidth]{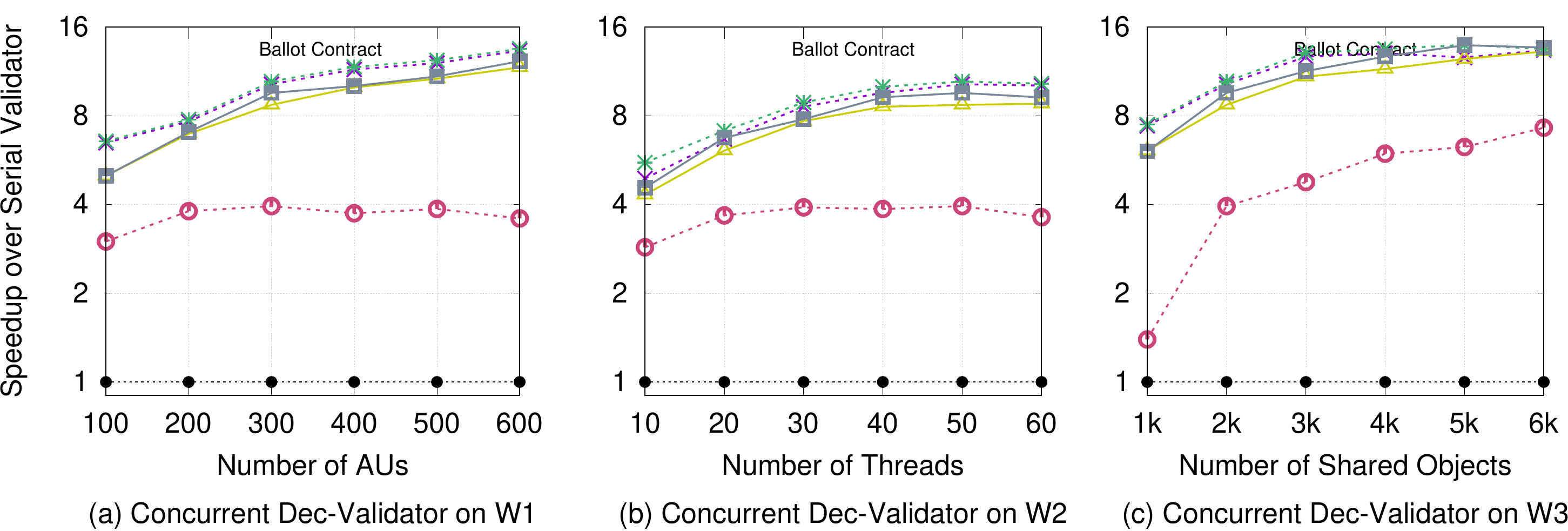}\vspace{-.4cm}
	\caption{Concurrent decentralized validator speedup over serial validator for ballot contract.}
	\label{fig:decval-speedup-ballot}
	
	
	\includegraphics[width=1\textwidth]{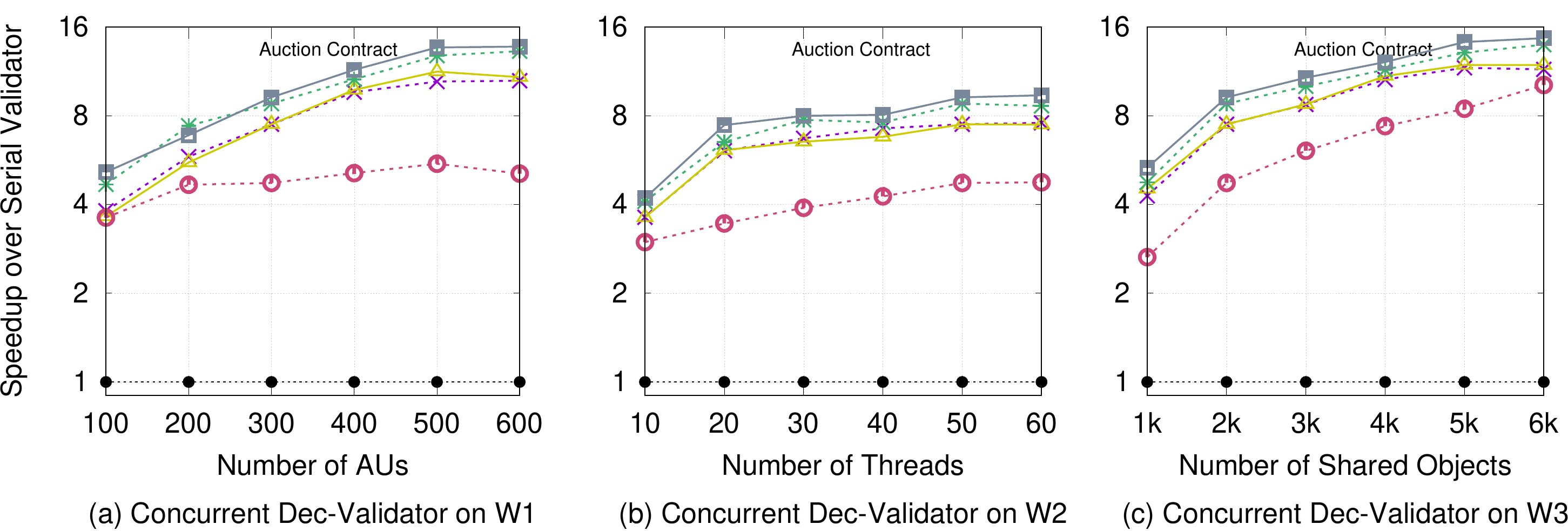}\vspace{-.4cm}
	\caption{Concurrent decentralized validator speedup over serial validator for auction contract.}
	\label{fig:decval-speedup-auction}
	
	
	\includegraphics[width=1\textwidth]{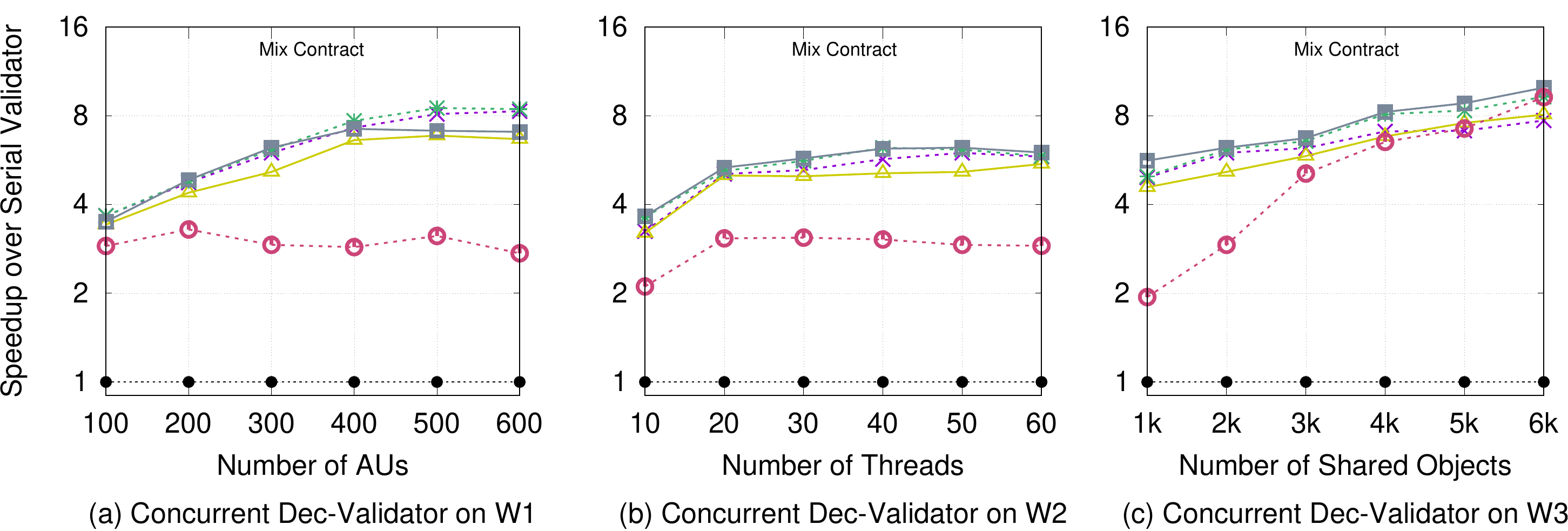}\vspace{-.4cm}
	\caption{Concurrent decentralized validator speedup over serial validator for mix contract.}
	\label{fig:decval-speedup-mix}
\end{figure}

\begin{figure}[H]
	\centering
		\includegraphics[width=1\textwidth]{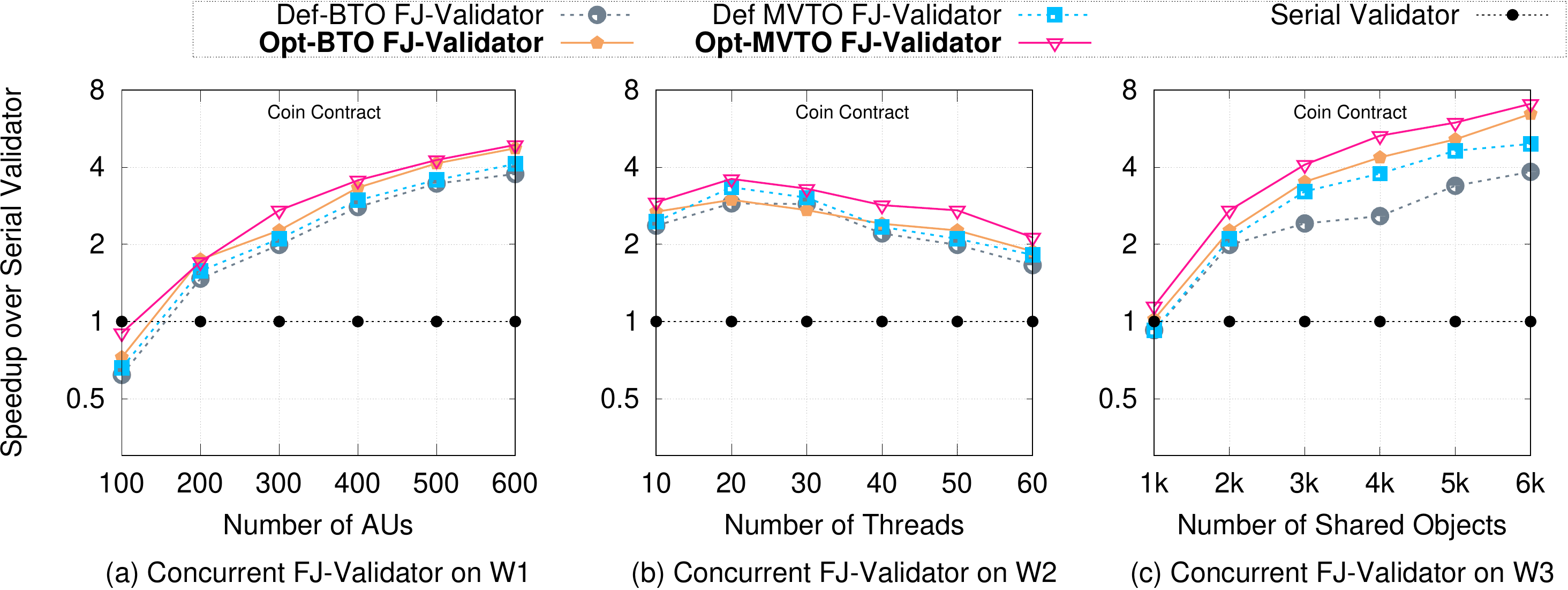}\vspace{-.4cm}
		\caption{Concurrent fork join validator speedup over serial validator for coin contract.}
		\label{fig:fj-speedup-coin}
		
		\vspace{.2cm}
		
		\includegraphics[width=1\textwidth]{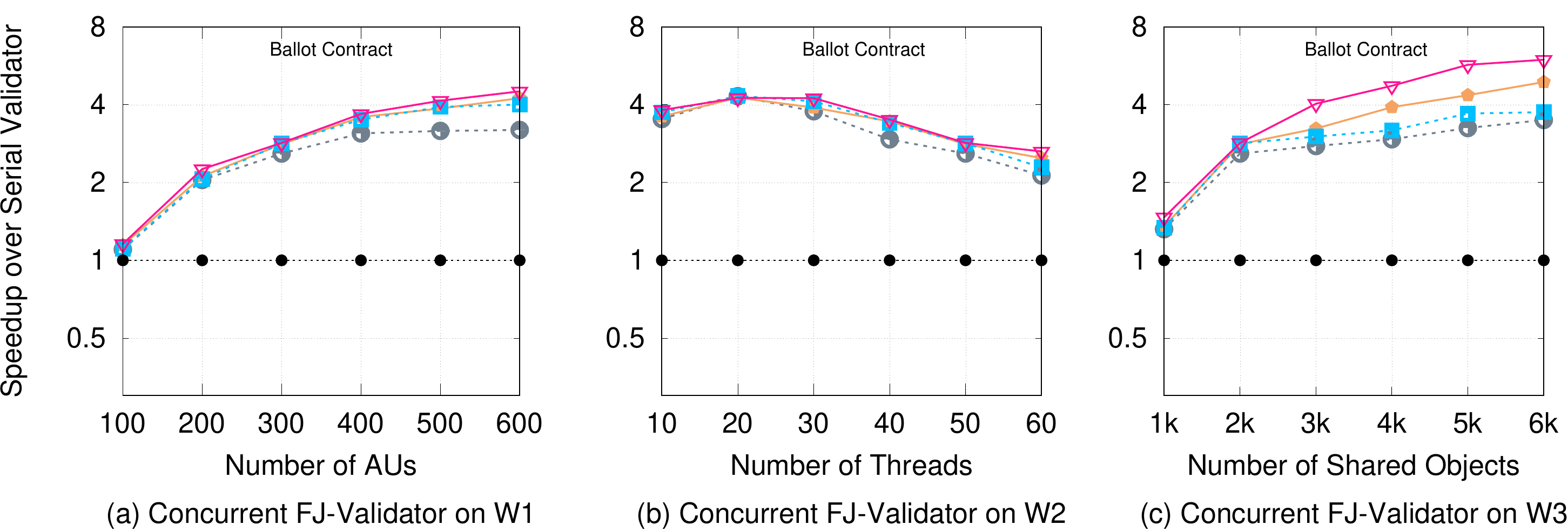}\vspace{-.4cm}
		\caption{Concurrent fork join validator speedup over serial validator for ballot contract.}
		\label{fig:fj-speedup-ballot}
		
		\vspace{.2cm}
		
		\includegraphics[width=1\textwidth]{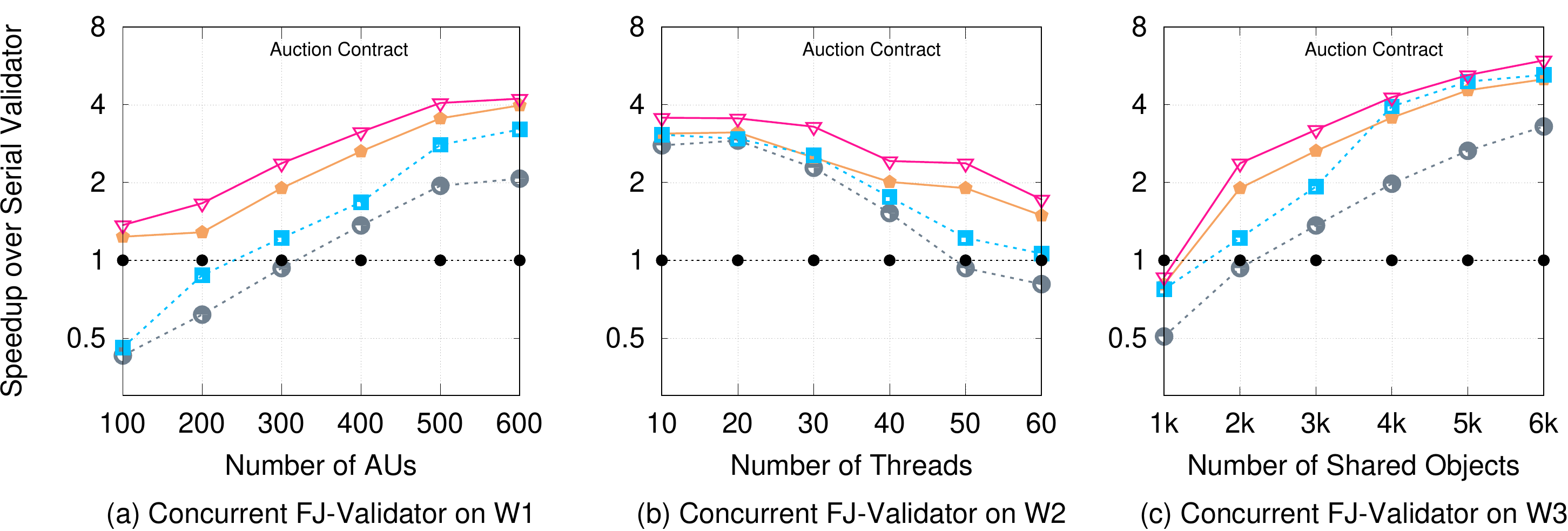}\vspace{-.4cm}
		\caption{Concurrent fork join validator speedup over serial validator for auction contract.}
		\label{fig:fj-speedup-auction}
		
		\vspace{.2cm}

		\includegraphics[width=1\textwidth]{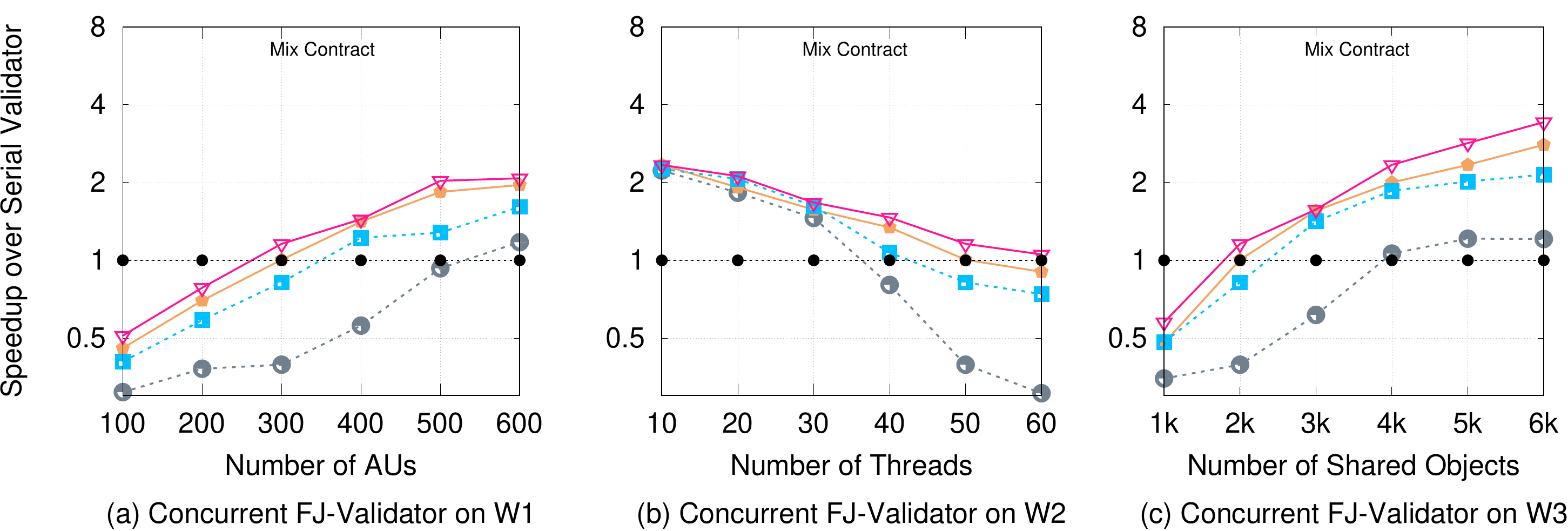}\vspace{-.4cm}
		\caption{Concurrent fork join validator speedup over serial validator for mix contract.}
		\label{fig:fj-speedup-mix}
\end{figure}

\figref{decval-speedup-coin}(a) through \figref{fj-speedup-mix}(a) show the speedup for concurrent validators over serial validators on W1. The speedup for concurrent validators (decentralized and fork-join) increases with an increase in \sctrn{s}. \figref{decval-speedup-coin}(a) to \figref{decval-speedup-mix}(a) demonstrate the speedup achieved by decentralized validator. It can be observed that for read-intensive benchmarks, the optimized MVTO decentralized validator (Opt-MVTO Dec-Validator) outperforms other validators. In contrast, in write-intensive benchmarks, the default MVTO decentralized validator (Def-MVTO Dec-Validator) achieves better speedup over other validators. Due to the overhead of multithreading for the concurrent bin with very fewer \sctrn{s}. We observed that with increasing \sctrn{s} in the blocks, the conflicts also increase. As a result, the number of transactions in the concurrent bin decreases. The speculative bin decentralized validator (Spec Bin Dec-Validator) speedup is quite less over concurrent STM Dec-Validators. Because STM miner precisely determines the dependencies between the \sctrn{s} of the block and harness the maximum concurrency than the bin-based miner. However, suppose the block consists of the \sctrn{s} with very few dependencies. In that case, Spec Bin Dec-Validator is expected to outperform other validators, as shown in the \figref{decval-speedup-coin}(a).


\figref{fj-speedup-coin}(a) to \figref{fj-speedup-mix}(a) show the speedup for fork-join validators on W1 for all the benchmarks. We can observe that the proposed optimized MVTO fork-join validator (Opt-MVTO FJ-Validator) outperforms other validators due to lower overheads at the fork-join master validator thread to allocate independent \sctrn{s} to slave validator threads. We noticed that decentralized concurrent validators speedup is quite high over fork-join concurrent validators because there is no bottleneck in this approach for allocating the \sctrn{s}. All threads in the decentralized approach work independently. It can also be observed that with fewer \sctrn{s} in several benchmarks, the speedup by fork-join validators drops to the point where it is less than the serial validators due to the overhead of thread creation dominate the speedup achieved, as shown in \figref{fj-speedup-ballot}(a), \figref{fj-speedup-auction}(a) and \figref{fj-speedup-mix}(a).

In W1, concurrent miners achieve a minimum of $\approx 2\times$ and maximum up to 10$\times$ speedup over serial miners across the contracts. The concurrent STM decentralized validators achieve speedup minimum $\approx 4\times$ and maximum up to $\approx 14\times$ while Spec Bin Dec-Validator ranges from $\approx 3\times$ to $\approx 9\times$ over serial miner across the contracts. The fork-join concurrent validators achieve a maximum speedup of $\approx 5\times$ over the serial validator.

\figref{miner-speedup-coin}(b) to \figref{fj-speedup-mix}(b) show the speedup on W2. The speedup increases with an increase in the number of threads. However, it is limited by the maximum number of logical threads in the experimental system. Thus, a slight drop in the speedup can be seen from 50 threads to 60 threads because the experimental system has a maximum of 56 logical threads. The reset of the concurrent miner observations is similar to the workload W1 based on read-intensive and write-intensive benchmarks. 

As shown in the \figref{decval-speedup-coin}(b) to \figref{decval-speedup-mix}(b), the concurrent decentralized validators speedup increase with an increase in threads. While as shown in \figref{fj-speedup-coin}(b) to \figref{fj-speedup-mix}(b), the concurrent fork-join validators speedup drops down with an increase in threads. The reason for this drop in the speedup is that the master validator thread in the fork-join approach becomes a bottleneck. The decentralized validator's observation shows that for the read-intensive benchmark, the Opt-MVTO Dec-validator outperforms other validators. While in the write-intensive benchmark, the Def-MVTO Dec-validator outperforms other validators, as shown in \figref{decval-speedup-ballot}(b). However, in the fork-join validator approach, the proposed Opt-MVTO FJ-validator outperforms all other validators due to the optimization benefit of bin based approach inclusion.


In W2, concurrent miners achieve a minimum of $\approx 1.5\times$ and achieves maximum up to $\approx 8\times$ speedup over serial miners across the contracts. The concurrent STM decentralized validators achieve speedup minimum $\approx 4\times$ and maximum up to $\approx 10\times$ while Spec Bin Dec-Validator ranges from $\approx 3\times$ to $\approx 7\times$ over serial miner across the contracts. The fork-join concurrent validators achieve a maximum speedup of $\approx 4.5\times$ over the serial validator.

The plots in \figref{miner-speedup-coin}(c) to \figref{fj-speedup-mix}(c) show the concurrent miners and validators speedup on W3. As shared objects increase, the concurrent miner speedup increases because conflict decreases due to less contention. Additionally, when contention is very low, more \sctrn{s} are added in the concurrent bin. However, it also depends on the contract. If the contract is a write-intensive, fewer \sctrn{s} are added in the concurrent bin. While more \sctrn{s} added in the concurrent bin for read-intensive contracts.

As shown in \figref{miner-speedup-coin}(c) and \figref{miner-speedup-mix}(c), the speculative bin miners surpass STM miners due to read-intensive contracts. While in \figref{miner-speedup-ballot}(c) and  \figref{miner-speedup-auction}(c), the Def-MVTO Miner outperform other miners as shared objects increase. In contrast, Def-BTO Miner performs better over other miners when \sctrn{s} are fewer because search time in write-intensive contracts to determine respective versions is much more in MVTO miner than BTO miner. Although, all concurrent miners performers better than the serial miner. In W3, concurrent miners start at around 1.3$\times$ and archives maximum up to 14$\times$ speedup over serial miners across all the contracts.


The speedup by validators (decentralized and fork-join) increases with shared objects. In \figref{decval-speedup-coin}(c), \figref{decval-speedup-auction}(c), and \figref{decval-speedup-mix}(c), proposed Opt-STM Dec-Validator perform better over other validators. However, for write-intensive contracts, the number of \sctrn{s} in the concurrent bin would be less. Therefore, the speedup by Def-STM Dec-Validators is greater than Opt-STM Dec-Validators, as shown in \figref{decval-speedup-ballot}(c). The Spec Bin Dec-Validator speedup is quite less over concurrent STM Dec-Validators because STM miner precisely determines the dependencies between the \sctrn{s} than the bin based miner.

In fork-join validators, proposed Opt-STM FJ-Validators outperform over all other FJ-Validators, as shown in \figref{fj-speedup-coin}(c) to \figref{fj-speedup-mix}(c) because of less contention at the master validator thread in the proposed approach to allocate independent \sctrn{s} to slave validator threads. We noticed that decentralized concurrent validators speedup is relatively high over fork-join concurrent validators with similar reasoning explained above. In W3, concurrent STM decentralized validators start at around 4$\times$ and achieve a maximum up to 14$\times$ speedup while Spec Bin Dec-Validator ranges from 1$\times$ to 14$\times$ speedup over serial miner across the contracts. The fork-join concurrent validators achieve a maximum speedup of 7$\times$ over the serial validator. The concurrent validators benefited from the work of the concurrent miners and outperformed serial validators.

\begin{figure}[!t]
	\centering
		{\includegraphics[width=1\textwidth]{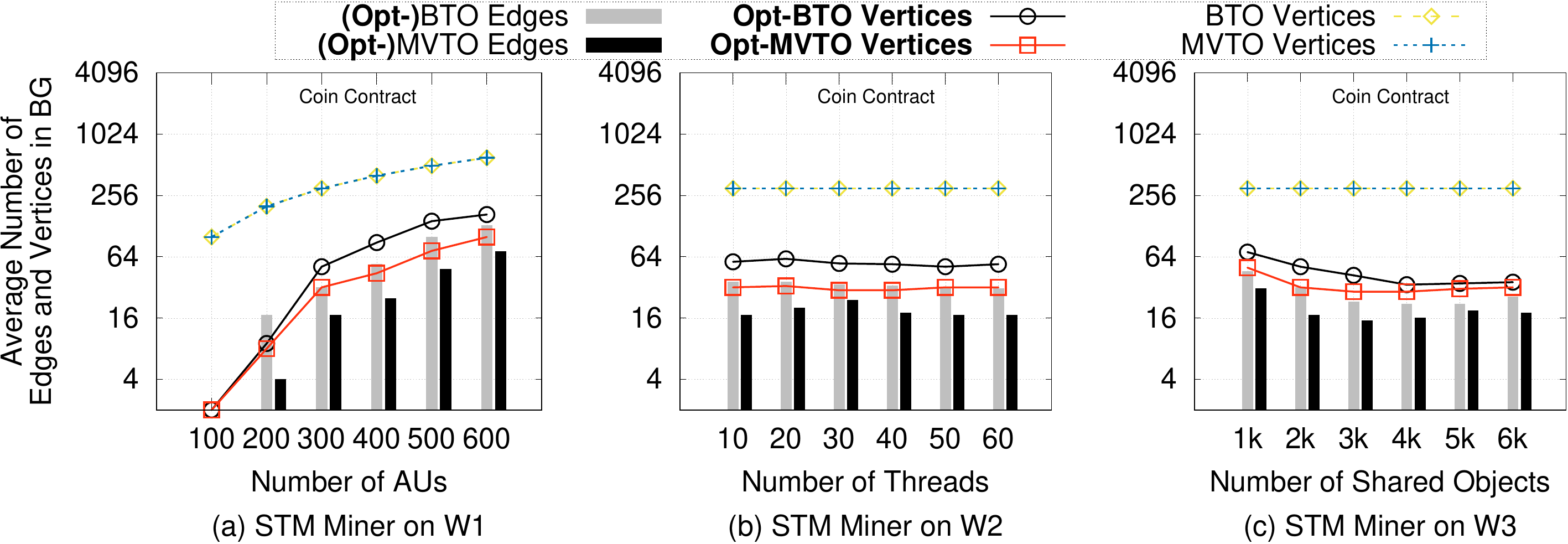}}\vspace{-.4cm}
		\caption{Average number of edges (dependencies) and vertices (\sctrn{s}) in block graph for coin contract.}
		\label{fig:bgCoin}
		
		\vspace{.2cm}    
		
		{\includegraphics[width=1\textwidth]{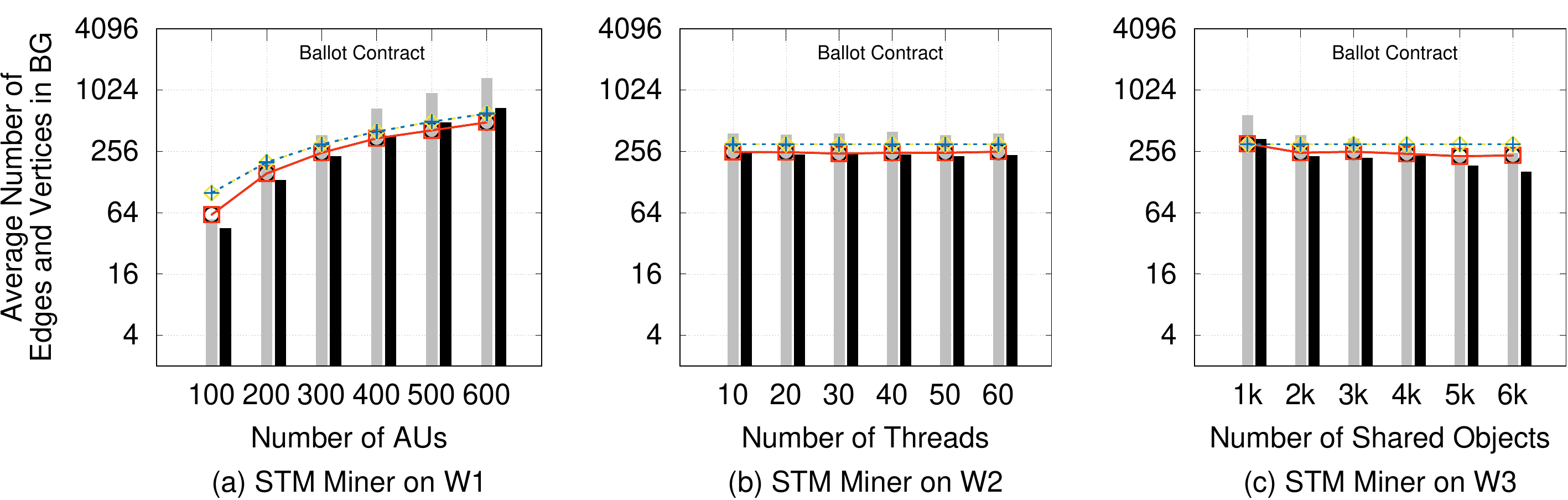}}\vspace{-.4cm}
		\caption{Average number of edges (dependencies) and vertices (\sctrn{s}) in block graph for ballot contract.}
		\label{fig:bgBallot}
		
		\vspace{.2cm}
		
		{\includegraphics[width=1\textwidth]{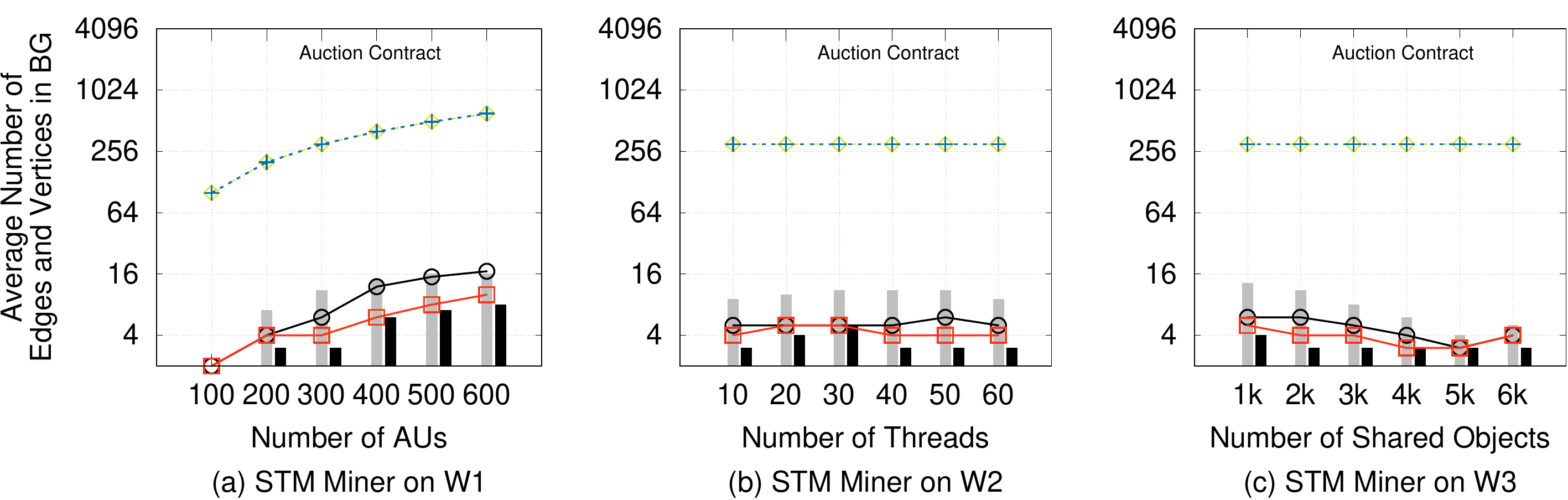}}\vspace{-.4cm}
		\caption{Average number of edges (dependencies) and vertices (\sctrn{s}) in block graph for auction contract.}
		\label{fig:bgAuction}
		
		\vspace{.2cm}

		{\includegraphics[width=1\textwidth]{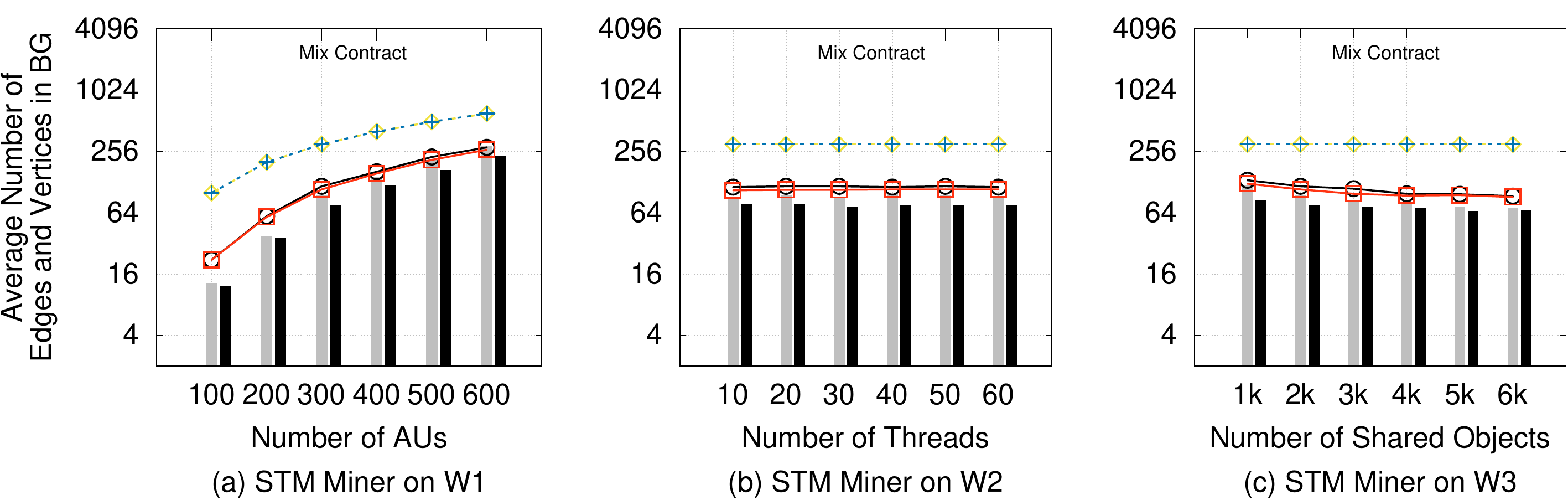}}\vspace{-.4cm}
		\caption{Average number of edges (dependencies) and vertices (\sctrn{s}) in block graph for mix contract.}
		\label{fig:bgMix}
\end{figure}



\figref{bgCoin} to \figref{bgMix} show the average number of edges (dependencies as histograms) and vertices (\sctrn{s} as line chart) in the BG for mix contract on all the workloads\footnote{We used histograms and line chart to differentiate vertices and edges to avoid confusion in comparing the edges and vertices.}. The average number of edges (dependencies) in the BG for both Default and Optimized approach for respective STM protocol remains the same; hence only two histograms are plotted for simplicity. As shown in the \figref{bgCoin}(a) to \figref{bgMix}(a) with increasing \sctrn{s} in W1, the BG edges and vertices also increase. It shows that the contention increases with increasing \sctrn{s} in the blocks. As shown in the \figref{bgCoin}(b) to \figref{bgMix}(b) in W2, the number of vertices and edges does not change much. However, in the W3, the number of vertices and edges decreases, as shown in \figref{bgCoin}(c) to \figref{bgMix}(c).


In our proposed approach, the BG consists of vertices respective to only conflicting \sctrn{s}, and non-conflicting \sctrn{s} are stored in the concurrent bin. While in Anjana et al.~\cite{Anjana+:CESC:PDP:2019} approach, all the \sctrn{s} had corresponding vertex nodes in the BG shown in \figref{bgCoin} to \figref{bgMix}. So, in W1, it will be 100 vertices in the BG if block consists of 100 \sctrn{s} and 200 if block consists of 200 \sctrn{s}. In W2 and W3, it will be 300 vertices. 
Having only conflicting \sctrn{s} vertices in BG saves much space because each vertex node takes 28-byte storage space.


The average block size in the Bitcoin and Ethereum blockchain is $\approx 1200$ KB~\cite{BitcoinAvgBlockSize} and $\approx 20.98$ KB~\cite{EthereumAvgBlockSize}, respectively measured for the interval of Jan 1$^{st}$, 2019 to Dec 31$^{th}$, 2020. Further, the block size keeps on increasing, and so the number of transactions in the block. The average number of transactions in the Ethereum block is $\approx 100$~\cite{EthereumAvgBlockSize}. Therefore, in the Ethereum blockchain, each transaction size is an average $\approx 0.2$ KB ($\approx 200$ bytes). We computed the block size based on these simple calculations when \sctrn{s} vary in the block for W1. The \eqnref{blocksize} is used to compute the block size (B) for the experiments. 
\begin{equation}
	B = 200 * N_{\sctrn{s}}
	\label{eq:blocksize}
\end{equation}
Where, $B$ is block size in bytes, $N_{\sctrn{s}}$ number of \sctrn{s} in block, and $200$ is the average size of an \sctrn{} in bytes.

To store the block graph $BG(V, E)$ in the block, we used \emph{adjacency list}. In the BG, a vertex node $V_s$ takes $28$ bytes storage, which consists of 3 integer variables and 2 pointers. While an edge node $E_{s}$ needs a total of $20$ bytes storage. The \eqnref{BGSize} is used to compute the size of BG ($\beta$ bytes). While \eqnref{perBG} is used to compute the additional space ($\beta_{p}$ percentage) needed to store BG in the block.

\begin{figure}[H]
	\centering
	{\includegraphics[width=1\textwidth]{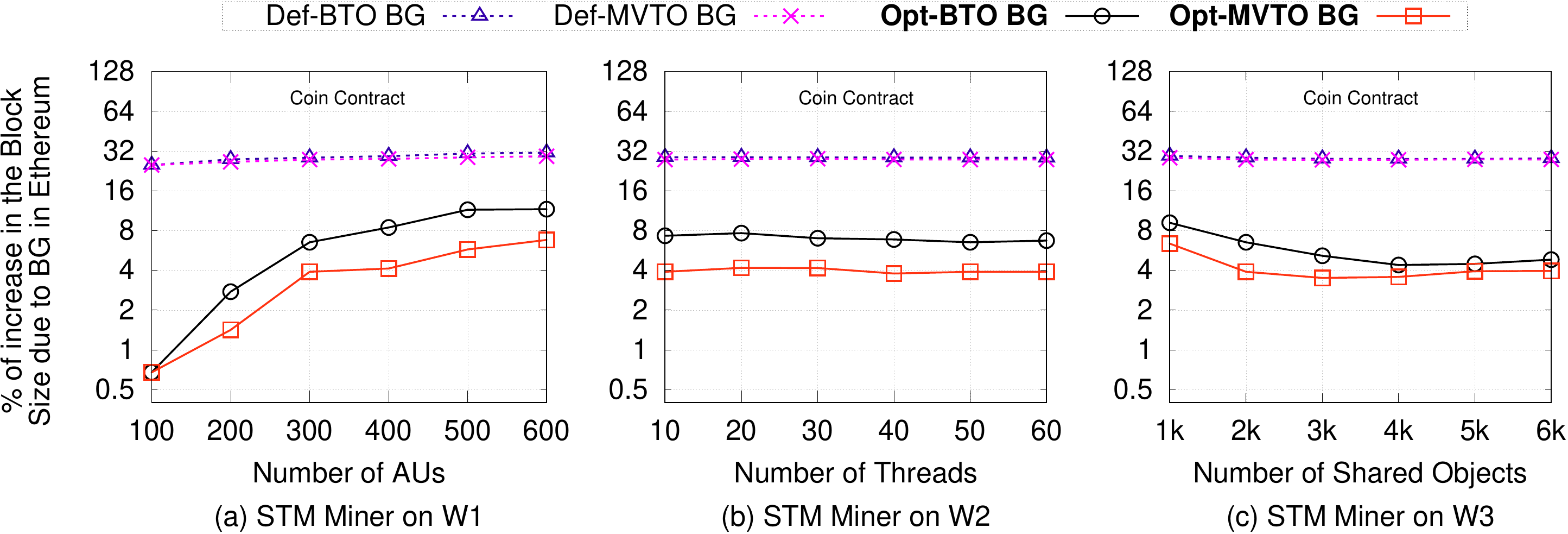}}\vspace{-.4cm}
	\caption{Percentage of additional space to store block graph in Ethereum block for coin contract.}
	\label{fig:incSizeCoin}
	
	\vspace{.22cm}    
	
	{\includegraphics[width=1\textwidth]{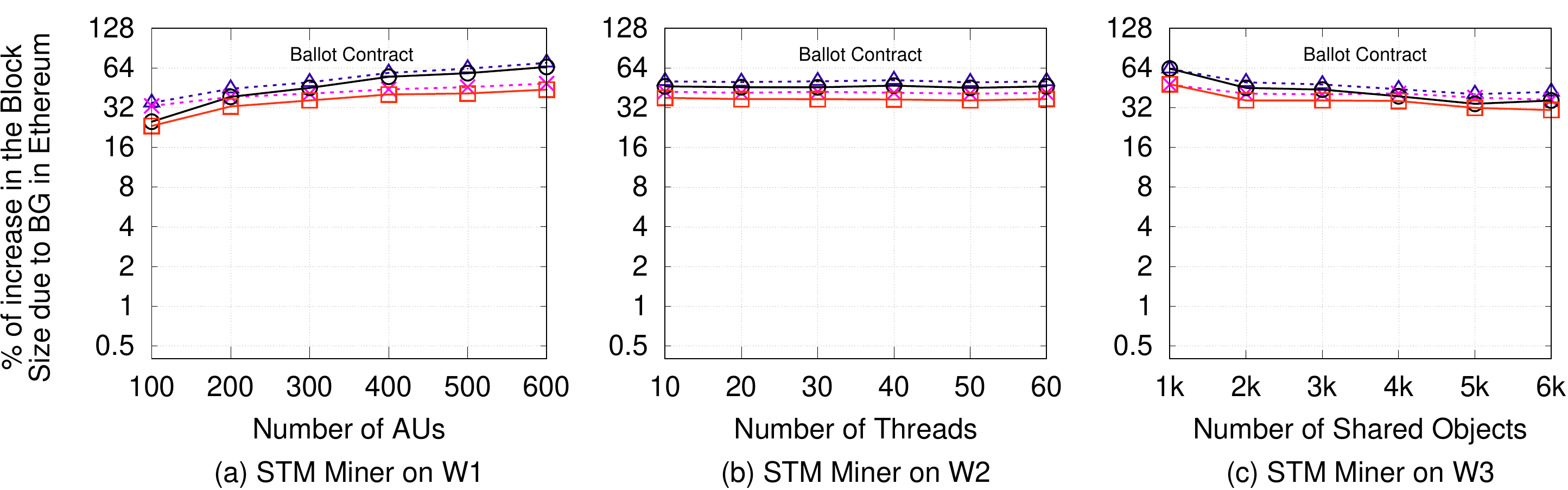}}\vspace{-.4cm}
	\caption{Percentage of additional space to store block graph in Ethereum block for ballot contract.}
	\label{fig:incSizeBallot}
	
	\vspace{.22cm}
	
	{\includegraphics[width=1\textwidth]{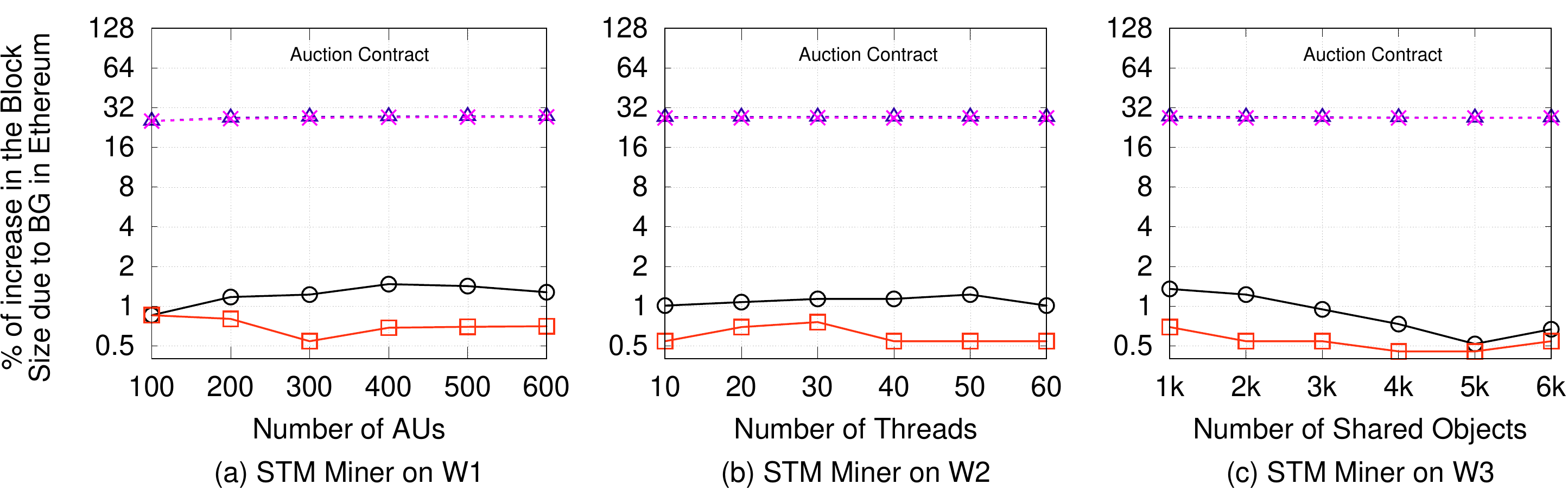}}\vspace{-.4cm}
	\caption{Percentage of additional space to store block graph in Ethereum block for auction contract.}
	\label{fig:incSizeAuction}
	
	\vspace{.22cm}
	
	{\includegraphics[width=1\textwidth]{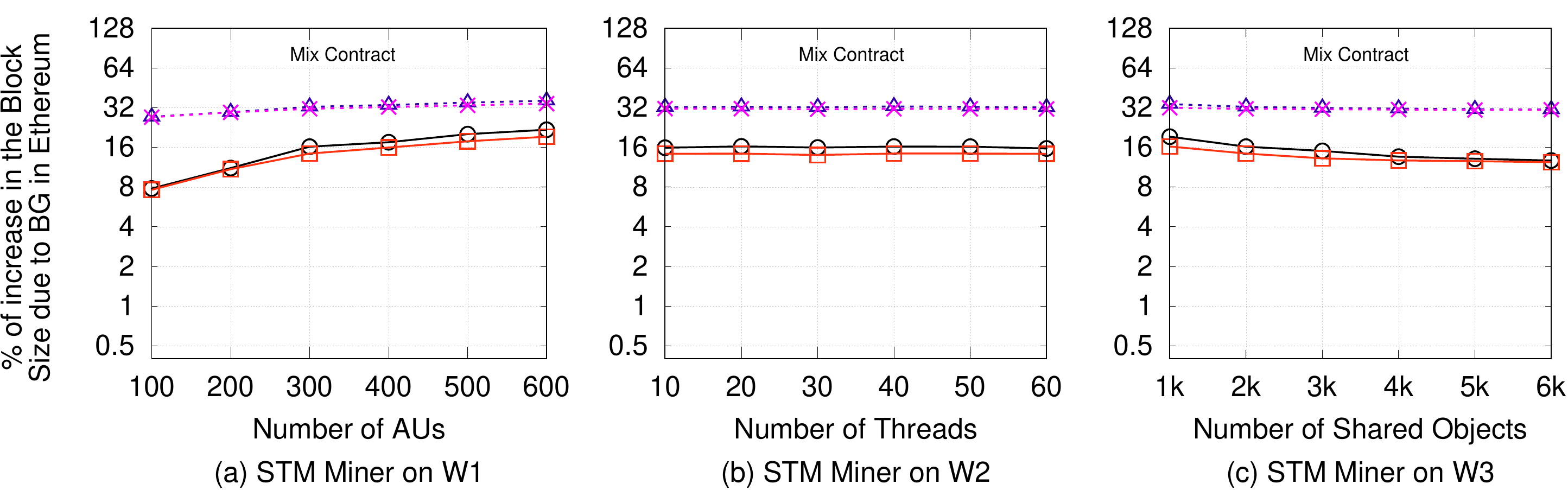}}\vspace{-.4cm}
	\caption{Percentage of additional space to store block graph in Ethereum block for mix contract.}
	\label{fig:incSizeMix}
\end{figure}

\begin{equation}
	\beta = (V_{s} * N_{\sctrn{s}}) + (E_{s} * M_{e})
	\label{eq:BGSize}
\end{equation}
Where, $\beta$ is size of BG in bytes, $V_s$ is size of a vertex node of $BG$ in bytes, $E_{s}$ is size (in bytes) of a edge node of $BG$, and $M_{e}$ is number of edges in $BG$.
\begin{equation}
	\beta_{p} = ({\beta*100})/{B}
	\label{eq:perBG}
\end{equation}


The plots in \figref{incSizeCoin} to \figref{incSizeMix} demonstrate the average percentage of additional storage space required to store BG in the Ethereum block on all benchmarks and workloads. We can observe that the space requirement also increases with an increase in the number of dependencies and vertices in BG. However, the space requirement of our proposed approach is smaller than the existing default approach. As shown in the \figref{bgBallot}, the dependencies and vertices are highest in the ballot contract compared to other contracts, so the space requirement is also high, as shown in \figref{incSizeBallot}. This is because the ballot is a write-intensive benchmark. It can be seen that the space requirements of BG by Opt-BTO BG and Opt-MVTO BG is smaller than Def-BTO BG and Def-MVTO BG miner, respectively.



The proposed approach significantly reduces the BG size for mix contract as shown in \figref{incSizeMix} across all the workloads. Which clearly shows the storage efficiency of the proposed approach. 
The storage advantage comes from using a bin-based approach combined with the STM approach, where concurrent bin information needs to be added into the block, which requires less space than having a corresponding vertex in BG for each \sctrn{s} of the block. So, we combine the advantages of both the approaches (STM and Bin) to get maximum speedup with storage optimal BG. {The average space required for BG in \% w.r.t. block size is $34.55\%$, $31.69\%$, $17.24\%$, and $13.79\%$ by Def-BTO. Def- MVTO, Opt-BTO, and Opt-MVTO approach, respectively. The proposed Opt-BTO and Opt-MVTO BG are $2\times$ (or $200.47\%$) and $2.30\times$ (or $229.80\%$) efficient over Def-BTO and Def-MVTO BG, respectively. With an average speedup of $4.49\times$ and $5.21\times$ for Opt-BTO, Opt-MVTO concurrent miner over serial, respectively. The Opt-BTO and Opt-MVTO decentralized concurrent validator outperform an average of $7.68\times$ and $8.60\times$ than serial validator, respectively.}
\section{Conclusion}
\label{sec:con}

To exploit the multi-core processors, we have proposed the concurrent execution of \scontract{} by miners and validators, which improves the throughput. Initially, the miner executes the smart contracts concurrently using optimistic STM protocol as BTO. To reduce the number of aborts and further improve efficiency, the concurrent miner uses MVTO protocol, which maintains multiple versions corresponding to each data object. Concurrent miner proposes a block that consists of a set of transactions, concurrent bin, BG, previous block hash, and the final state of each shared data objects. Later, the validators re-execute the same \SContract{} transactions concurrently and deterministically in two-phase using concurrent bin followed by the BG given by miner, which capture the conflicting relations among the transactions to verify the final state. 
{Overall, the proposed Opt-BTO and Opt-MVTO BG are $2\times$ (or $200.47\%$) and $2.30\times$ (or $229.80\%$) efficient over Def-BTO and Def-MVTO BG, respectively. With an average speedup of $4.49\times$ and $5.21\times$ for Opt-BTO, Opt-MVTO concurrent miner over serial, respectively. The Opt-BTO and Opt-MVTO decentralized concurrent validator outperform an average of $7.68\times$ and $8.60\times$ than serial validator, respectively.}

\vspace{1mm}
\noindent
\textbf{Acknowledgements.} This project was partially supported by a research grant from Thynkblynk Technologies Pvt. Ltd, and MEITY project number 4(20)/2019-ITEA. 

\bibliography{citations}

\end{document}